\documentclass[final, twocolumn]{IEEEtran}

\usepackage{hyperref}
\hypersetup{citecolor=blue,pdfborder=0 0 0,pdfstartview={FitH}}
\newcommand\Tstrut{\rule{0pt}{2.5ex}}         
\newcommand\Bstrut{\rule[-1.1ex]{0pt}{0pt}}
\usepackage{ifpdf}
\usepackage{cite}
\usepackage{tabularx}

\ifCLASSINFOpdf
 \usepackage[pdftex]{graphicx}

\usepackage{color}

\else

\fi

\usepackage{subcaption}
\usepackage[cmex10]{amsmath}
\usepackage{algorithmic}
\usepackage{array}
\usepackage{fixltx2e}
\usepackage{stfloats}
\newtheorem{theorem}{Theorem}
\newtheorem{lemma}[theorem]{Lemma}
\newtheorem{example}{Example}

\newenvironment{proof}[1][Proof]{\begin{trivlist}
\item[\hskip \labelsep {\bfseries #1}]}{\end{trivlist}}

\newcommand{\qed}{\nobreak \ifvmode \relax \else
      \ifdim\lastskip<1.5em \hskip-\lastskip
      \hskip1.5em plus0em minus0.5em \fi \nobreak
      \vrule height0.75em width0.5em depth0.25em\fi}
\begin{document}
\title{Iterative Decoding of LDPC Codes over the $q$-ary Partial Erasure Channel}
\author{Rami~Cohen,~\IEEEmembership{Graduate Student Member,~IEEE,}        
        and~Yuval~Cassuto,~\IEEEmembership{Senior Member,~IEEE}
\thanks{The authors are with the Department of Electrical Engineering,
Technion - Israel Institute of Technology, Haifa, Israel 3200003 (email: rc@campus.technion.ac.il, ycassuto@ee.technion.ac.il)}
\thanks{This work was supported in part by the German-Israel Foundation, by
the Israel Ministry of Science and Technology, and by the Israel Science
Foundation.}
\thanks{Parts of this work were presented at the 2014 IEEE International Symposium on Information Theory (ISIT), Hawaii, USA, at the 8th International Symposium on Turbo Codes \& Iterative Information Processing (ISTC), Bremen, Germany, and at the 2015 IEEE International Workshop on Information Theory (ITW), Jerusalem, Israel.}
\thanks{ Copyright (c) 2014 IEEE. Personal use of this material is permitted.  However, permission to use this material for any other purposes must be obtained from the IEEE by sending a request to pubs-permissions@ieee.org.}}

\markboth{}%
{Shell \MakeLowercase{\textit{et al.}}:Low-Density Parity-Check Codes over the $q$-ary Partial Erasure Channel }

\maketitle

\begin{abstract}
In this paper, we develop a new channel model, which we name the $q$-ary partial erasure channel (QPEC). The QPEC has a $q$-ary input, and its output is either the input symbol or a set of $M$ ($2 \le M \le q$) symbols, containing the input symbol. This channel serves as a generalization to the binary erasure channel, and mimics situations when a symbol output from the channel is known only partially; that is, the output symbol contains some ambiguity, but is not fully erased. This type of channel is motivated by non-volatile memory multi-level read channels. In such channels the readout is obtained by a sequence of current/voltage measurements, which may terminate with partial knowledge of the stored level. Our investigation is concentrated on the performance of low-density parity-check (LDPC) codes when used over this channel, thanks to their low decoding complexity using belief propagation. We provide the exact QPEC density-evolution equations that govern the decoding process, and suggest a cardinality-based approximation as a proxy. We then provide several bounds and approximations on the proxy density evolutions, and verify their tightness through numerical experiments. Finally, we provide tools for the practical design of LDPC codes for use over the QPEC.
\end{abstract}


\IEEEpeerreviewmaketitle

\section{Introduction}


The rapid development of memory technologies have introduced challenges to the continued scaling of memory devices in density and access speed. One of the common computer memory technologies is non-volatile memory (NVM). In multi-level NVMs, such as flash memories, an information symbol is represented in a memory cell by one of $q$ voltage levels, where information is written/read by adding/measuring cell voltage \cite{MLC, Gal}. The read process is usually performed by comparing the stored voltage level to certain threshold voltage levels. To scale storage density in NVMs, the number of levels per cell is continuously increased \cite{Meena, flash_4bits}.  As the number of levels increases, errors become more and more prevalent due to intercell interference \cite{Huang1}. In addition, the use of multi-level memory cells in the  emerging technology of resistive memories introduces significant reliability challenges \cite{ReRAM}.

Apart from classical channels and error models, multi-level memories motivate coding for a diversity of new channels with rich features. Our work here is motivated by a class of channels we call \textit{measurement channels}, in which information is written/read by adding/measuring electrical charges. These channels encompass a variety of equivocations introduced to the information by an imperfect read process, due to either physical limitations or speed constraints. The channel we propose and study here -- the \textit{$q$-ary partial erasure channel} (QPEC) -- is a basic and natural model for a measurement channel in multi-level memories. The model comes from a read process that occasionally fails to read the information at its entirety, and provides as decoder inputs $q$-ary symbols that are \textit{partially} erased. In the QPEC model, an output symbol is a \textit{set} of symbols that includes the correct input symbol. This set can be either of cardinality $1$ or a set of $M$ symbols ($M$ is a parameter, $2 \le M \le q$). In the latter case, we say that a \textit{partial erasure} event occurred, modeling the uncertainty that may occur in the read process due to imperfect measurements. Theoretically speaking, the QPEC is a generalization of the binary (or $q$-ary) erasure channel (BEC or QEC). In the BEC/QEC, symbols are either received perfectly or erased completely; in the QPEC, partially erased symbol provide information in quantity that grows as $M$ gets smaller.

In this work, we suggest the use of GF($q$) low-density parity-check (LDPC) codes \cite{Gallager1, Gallager2, Davey} for encoding information over the QPEC, due to their low complexity of implementation and good performance under iterative decoding \cite{mct}. These codes were shown to achieve performance close to the capacity for several important channels, using efficient decoding algorithms \cite{mct, Chung_cap}. Non-binary LDPC codes were considered in several works, such as in \cite{Davey, Song, Bennatan, Rathi}, and were shown as superior to binary codes in several cases \cite{Davey, Rathi}. In our analysis, we propose a message-passing decoder for decoding GF($q$)-LDPC codes over the QPEC, extending the known iterative decoder for the BEC/QEC to deal with partial erasures. The iterative-decoding performance evaluation of LDPC codes is usually performed using the density-evolution method \cite{RU} that tracks the decoding failure probability. However, this method becomes prohibitively complex in practice as $q$ increases, as it requires an iterative evaluation of multi-dimensional probability distributions \cite{RU, Bennatan}. Thus, we provide approximation schemes for tracking the QPEC decoding performance efficiently and verify their tightness. Finally, we develop tools for the design of good LDPC codes for the QPEC.

The paper is structured as follows. We begin by introducing the QPEC channel and its capacity in Section \ref{QPEC_DEF}. In Section \ref{iterative_decoding}, we give a short review of GF($q$)-LDPC codes and propose a message-passing based decoder for the QPEC. Decoding performance analysis is provided in Sections \ref{de_analysis} and \ref{pm_approx}, and code design tools are discussed in Section \ref{code_design_considerations}. Finally, conclusions are given in Section \ref{conc}.

\section{The $q$-ary Partial Erasure Channel}
\label{QPEC_DEF}
\subsection{Channel model}
\label{channel_model}

The {\it $q$-ary partial erasure channel} (QPEC) is a generalization of the well known binary erasure channel (BEC) \cite{Elias} in two ways. First, similarly to the $q$-ary erasure channel (QEC), its input alphabet is $q$-ary, with $q \geq 2$. Second, generalizing the BEC erasure event, a partial erasure occurs when the input symbol is known to belong a set of $M$ ($M$ is a parameter, $2 \le M \le q$) symbols (rather than $q$ symbols). The QPEC is defined as follows. Let $X$ be the transmitted symbol, taken from the alphabet ${\cal X} = \left\{ {0,{\alpha ^0}=1,{\alpha ^1},...,{\alpha ^{q - 2}}} \right\}$ of cardinality $q$, where $\alpha$ is a primitive element of the finite field GF($q$) (i.e., the elements in $\mathcal{X}$ are the field elements). Define the super-symbols ${?_x^{\left( i \right)}}$ for each $x \in \mathcal{X}$ and for $i=1,2,...,{q-1 \choose M-1}$, such that ${?_x^{\left( i \right)}}$ is a set of $M$ symbols containing the symbol $x$ and $M-1$ other symbols, taken from $\mathcal{X}\backslash \left\{ x \right\}$. The output $Y$ (given an input symbol $x$) is a \textit{set} of symbols, which is either the singleton $\left\{ x \right\}$, or one of the sets ${?_x^{\left( i \right)}}$ (of cardinality $M$) for some $i$. Therefore, the output alphabet $\mathcal{Y}$ contains all possible sets of cardinality $1$ and cardinality $M$ taken from $\mathcal{X}$. The transition probabilities governing the QPEC are as follows:
\begin{equation}
\label{tran_matrix}
\Pr \left( {\left. {Y = y} \right|X = x} \right) = \left\{ {\begin{array}{*{20}{c}}
{1 - \varepsilon ,}&{y = \left\{ x \right\}}\\
{\varepsilon /{q-1 \choose M-1},}&{y = ?_x^{\left(i \right)}},
\end{array}} \right.
\end{equation}
where $0 \le \varepsilon  \le 1$ is the (partial) erasure probability. 

That is, with probability $1-\varepsilon$ the input symbol is received with no error, and with probability $\varepsilon$ a partial-erasure event occurs, such that the input symbol is known up to $M$ symbols. In the latter case, the output sets ${?_x^{\left( i \right)}}$ are equiprobable. This models a situation of maximum uncertainty at the output, which is uniformly distributed on sets of cardinality $M$ containing $x$. Note that for $M=q=2$ the QPEC is equivalent to the BEC, where for $M=q>2$ the QPEC is equivalent to the QEC. The transition probabilities are given explicitly in the following example for a particular choice of $q,M$ and a transmitted symbol $x$.
\begin{example} Assume that $q=4$ and that the symbol $0$ was transmitted. If $M=2$, the possible output sets, and their transition probabilities from \eqref{tran_matrix}, are given by:
\begin{equation}
\Pr {\left( {\left. {Y = y} \right|X = 0} \right)_{q = 4,M = 2}} = \left\{ {\begin{array}{*{20}{c}}
{1 - \varepsilon ,}&{y = \left\{ 0 \right\}}\\
{\varepsilon /3,}&{y = \left\{ {0,{\alpha ^0}} \right\}}\\
{\varepsilon /3,}&{y = \left\{ {0,{\alpha ^1}} \right\}}\\
{\varepsilon /3,}&{y = \left\{ {0,{\alpha ^2}} \right\}.}
\end{array}} \right.
\end{equation}
\end{example}

\subsection{Capacity}
Denote ${p_x} \buildrel \Delta \over = \Pr \left( {X = x} \right)$, for $x=0,\alpha^0,\alpha^1...,\alpha^{q-2}$, to be the input distribution to the channel. According to the definition of the channel capacity $C$,
\begin{equation}
\label{cap_def}
C = {\max _{\left\{ {{p_x}} \right\}}} \hspace{4pt} I\left( {X;Y} \right) = {\max _{\left\{ {{p_x}} \right\}}}\left( {H\left( Y \right) - H\left( {\left. Y \right|X} \right)} \right),
\end{equation}
where $I\left( {X;Y} \right)$ is the mutual information between the input $X$ and the output $Y$, and $H\left( Y \right)$, ${H\left( {\left. Y \right|X} \right)}$ are the entropy of $Y$ and the conditional entropy of $Y$ given $X$, respectively. The conditional entropy ${H\left( {\left. Y \right|X} \right)}$ can be calculated using \eqref{tran_matrix}:
\begin{equation}
\label{cond_entropy}
H\left( {\left. Y \right|X} \right) =  - \left( {1 - \varepsilon } \right)\log \left( {1 - \varepsilon } \right) - \varepsilon \log \left( {\varepsilon /{q-1 \choose M-1}} \right).
\end{equation}
The conditional entropy is independent of ${\left\{ {{p_x}} \right\}}$ (as expected), implying that it is sufficient to maximize the entropy ${H\left( Y \right)}$ to find the capacity. The QPEC capacity is provided in the following theorem.

\begin{theorem}
\label{QPEC_cap}
\emph{(Capacity)} The QPEC capacity is:
\begin{equation}
\label{QPEC_capacity}
C\left( {{\text{QPEC}}} \right) = 1 - \varepsilon {\log _q}M,
\end{equation}
measured in $q$-ary symbols per channel use.
\end{theorem}
The proof of this theorem is provided in Appendix \ref{proof:uniform_dist}. As one may expect due to the uniform distribution of the output when a partial-erasure occurs, ${H\left( Y \right)}$ is maximized under the uniform distribution of the input (i.e., for $p_x=1/q$). Note the agreement of \eqref{QPEC_capacity} with the QEC capacity for $M=q$, and in particular with the BEC capacity for $M=q=2$.

\subsection{Maximum-likelihood decoding}
\label{section:ML_decoding}

Assume that a codeword $\boldsymbol{c}$ taken from a codebook $\mathcal{C}$ was transmitted over the QPEC and that the output $\boldsymbol{y}$ was received. The elements $y_i$ of $\boldsymbol{y}$ should be understood in a generalized sense, as they contain either a set of one symbol or a set of $M$ symbols (according to the transition probabilities in Equation \eqref{tran_matrix}). For $M=q$, in which the QPEC is essentially the QEC, codewords coinciding with $\boldsymbol{y}$ in non-erased positions are said to be \textit{compatible} with $\boldsymbol{y}$ \cite{mct}, and they serve as maximum-likelihood (ML) decoding of $\boldsymbol{y}$. However, when $M<q$, partially-erased codeword symbol positions should be considered for the ML decoding of $\boldsymbol{y}$.

To extend the notion of compatibility to QPECs with $M<q$, we define the set:
\begin{equation}
\Psi  = \left\{ {{\boldsymbol{c}} \in \mathcal{C}:\forall i,{c_i}\bigcap {{y_i} \neq \emptyset } } \right\},
\end{equation}
which is the set of all codewords that have in each position a symbol that is contained in the corresponding output of $\boldsymbol{y}$ in the same position. Each codeword in $\Psi$ can serve as an ML decoding of $\boldsymbol{c}$, since $\boldsymbol{c}$ and $\boldsymbol{y}$ must agree in non-erased positions, and in the remaining positions the correct transmitted codeword symbol $c_i$ is contained in $y_i$ by the QPEC definition. Therefore, $\boldsymbol{y}$ is decoded correctly (with probability $1$) if and only if $\left| \Psi  \right| = 1$. In a similar manner, when ML \textit{symbol} decoding is used, $y_i$ is decoded correctly (with probability $1$) if and only if all the codewords in $\Psi$ contain the same symbol in their $i^{\text{th}}$ position. In practice, ML decoding complexity is usually prohibitive. In the next section, we move to specify a low-complexity iterative message-passing decoder for GF($q$) LDPC codes used over the QPEC.

\section{GF($q$) LDPC Codes and Message-Passing Decoding}
\label{iterative_decoding}
\subsection{GF($q$) LDPC codes}
\label{LDPC_Q}
Before developing our coding results for the QPEC, we include some well-known facts on LDPC codes as a necessary background. A GF($q$) $[n,k]$ LDPC code is defined in a similar way to its binary counterpart, by a sparse parity-check matrix, or equivalently by a Tanner graph \cite{tanner}. This graph is bipartite, with $n$ variable (left) nodes, which correspond to codeword symbols, and $n-k$ check (right) nodes, which correspond to parity-check equations. The codeword symbols are taken from GF($q$), where the labels on the graph edges are taken from the non-zero elements of GF($q$). In the graph, a check node $\mathtt{c}$ is connected by edges to variable nodes $\mathtt{v} \in \mathcal{N}(\mathtt{c})$, where $\mathcal{N}(\mathtt{c})$ denotes the set of variable nodes adjacent to check node $\mathtt{c}$. The induced parity-check equation is $\sum\limits_{\mathtt{v} \in N\left( \mathtt{c} \right)} {{h_{\mathtt{c},\mathtt{v}}}\cdot {\mathtt{v}}}  = 0$, where $h_{\mathtt{c},\mathtt{v}}$ are the labels on the edges connecting variable node $\mathtt{v}$ to check node $\mathtt{c}$. Note that the calculations are performed using GF($q$) arithmetic. An example of a Tanner graph is given in Figure \ref{fig_sim}.

\begin{figure}[!t]
\centering
\includegraphics[width=2.5in]{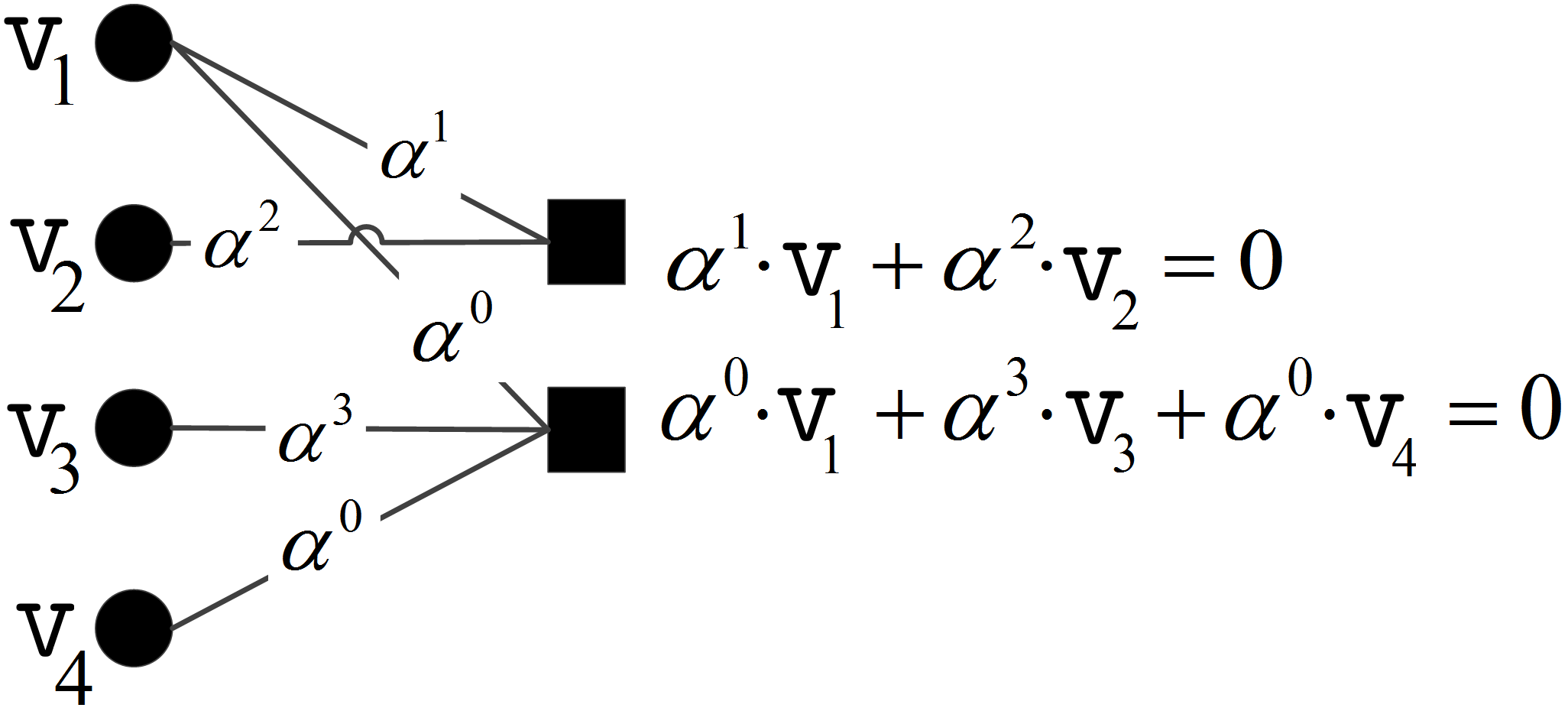}
\caption{An example of a Tanner graph over GF($4$). Circles denote variable nodes (codeword symbols), and squares denote check nodes (parity-check equations). The symbols on the edges are the labels, leading to the parity-check equations on the right.}
\label{fig_sim}
\end{figure}

LDPC codes are usually characterized by the {\it degree distributions} of the variable nodes and the check nodes. They are called {\it regular} if both variable nodes and check nodes have constant degree. Otherwise, they are called \textit{irregular}. Denote by $d_v$ and $d_c$ the maximal degree of variable nodes and check nodes, respectively. As is customary \cite{mct}, we define the following degree-distribution polynomials:
\begin{equation}
\label{lambda_polynomial}
{\lambda \left( x \right) = \sum\limits_{i = 2}^{{d_v}} {{\lambda _i}{x^{i - 1}}} },
\end{equation}
\begin{equation}
\label{rho_polynomial}
{\rho \left( x \right) = \sum\limits_{i = 2}^{{d_c}} {{\rho _i}{x^{i - 1}}} },
\end{equation}
where for each $i$, a fraction $\lambda_i$ ($\rho_i$) of the edges is connected to variable (check) nodes of degree $i$. These polynomials will be used later for analyzing the iterative-decoding performance of LDPC codes over the QPEC. The {\it design rate} $r$ of an LDPC code with degree-distribution polynomials $\lambda(x)$ and $\rho(x)$, measured in $q$-ary symbols per channel use, is \cite{mct}:
\begin{equation}
\label{LDPC_rate}
r = 1 - \frac{{\int\limits_0^1 {\rho \left( x \right)} dx}}{{\int\limits_0^1 {\lambda \left( x \right)dx} }} = 1 - \frac{{\sum\limits_{i = 2}^{{d_c}} {{\rho _i}/i} }}{{\sum\limits_{i = 2}^{{d_v}} {{\lambda _i}/i} }}.
\end{equation}
The design rate equals to the actual rate if the rows of the LDPC code parity-check matrix are linearly independent.

\subsection{Message-passing decoder for the QPEC}
\label{message_passing}

The following decoder for GF($q$) LDPC codes over the QPEC is a variation of the standard message passing/belief propagation algorithm over a Tanner graph, generalizing the iterative decoding process used over the BEC/QEC. The key change is that in the QPEC setting the exchanged beliefs are sets of symbols, rather than individual symbols (and erasure symbols) as with the BEC/QEC. We have two types of messages at each decoding iteration $l$: \textit{variable to check} (VTC) messages and \textit{check to variable} (CTV) messages, denoted ${{\rm{VTC}}_{\mathtt{v} \to \mathtt{c}}^{(l)}}$ and ${{\rm{CTV}}_{\mathtt{c} \to \mathtt{v}}^{(l)}}$, respectively. Each outgoing message from a variable (check) node to a check (variable) node depends on all its incoming messages, except for the incoming message originated from the target node. At iteration $l=0$, channel information is sent from variable nodes to check nodes: partially-erased nodes send sets of symbols of cardinality $M$, while non-erased ones send sets of cardinality $1$ (recall that both sets contain the correct symbol). The channel information sent from variable node $\mathtt{v}$ will be denoted ${{\rm{VTC}}_{\mathtt{v}}^{(0)}}$. 

\begin{figure*}[t!]
    \centering
    \begin{subfigure}[t]{0.45\textwidth}
        \centering
        \includegraphics[height=1.2in]{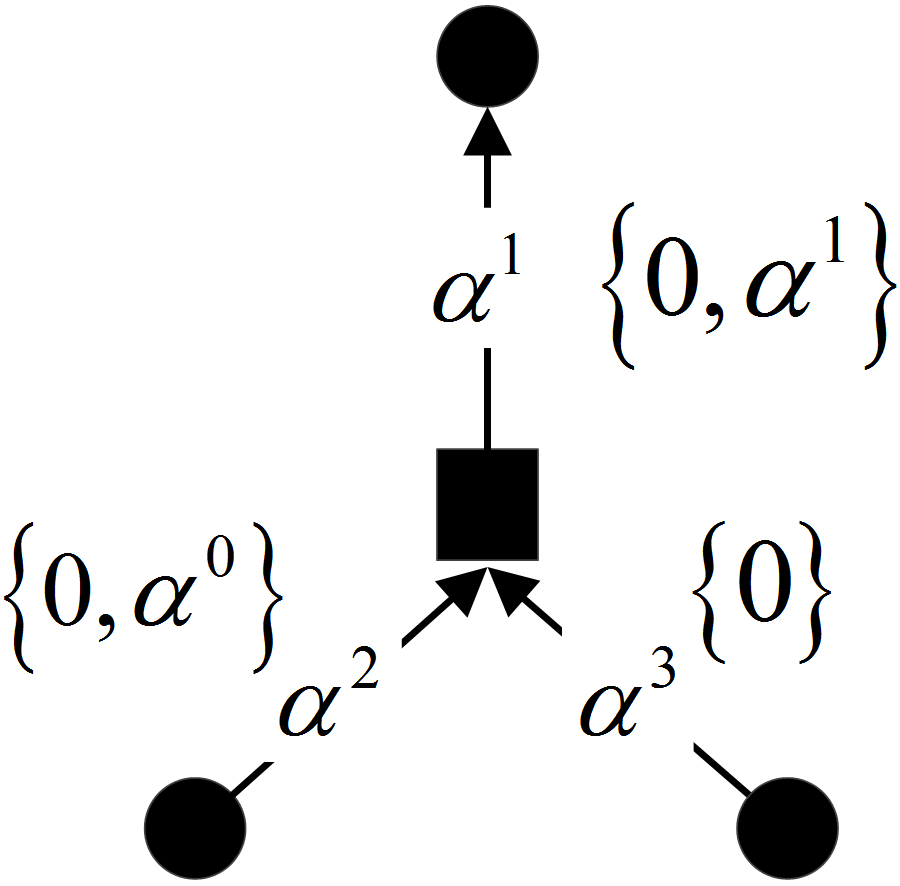}
        \caption{An example for a CTV message over GF($4$). The outgoing message is ${{{\alpha ^2}} \over \alpha } \cdot \left\{ {0,{\alpha ^0}} \right\} + {{{\alpha ^3}} \over \alpha } \cdot \left\{ 0 \right\} = \left\{ {0,{\alpha ^1}} \right\}$.}
\label{fig_CTV}
    \end{subfigure}%
    ~ \hspace{5pt}
    \begin{subfigure}[t]{0.45\textwidth}
       \centering
        \includegraphics[height=1.2in]{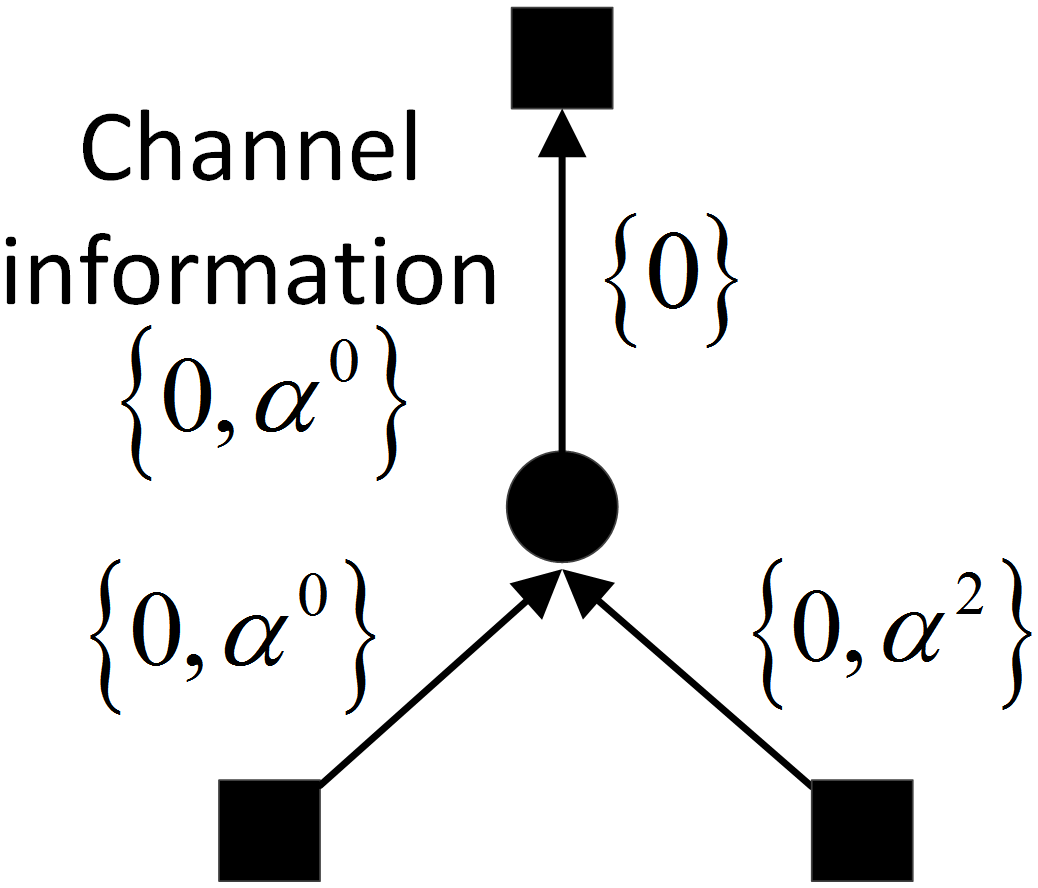}
        \caption{An example for a VTC message. The edge labels are not shown as they are not required for the VTC message calculation.}
\label{fig_VTC}
    \end{subfigure}
    \caption{Examples for CTV/VTC messages in the decoding process. Circles denote variable nodes, and squares denote check nodes. Symbols on edges represent edge labels. The two sets on the bottom are incoming messages, where the set on the top is the corresponding outgoing message.}
\end{figure*}

In the subsequent iterations, the operations of the message-passing decoder translate to the following operations on sets. ${{\rm{CTV}}_{\mathtt{c} \to \mathtt{v}}^{(l)}}$ contains the possible values of $\mathtt{v}$ given the incoming messages from the variable nodes $\left\{ {\mathcal{N}\left( \mathtt{c} \right)\backslash \mathtt{v}} \right\}$, such that the parity-check equation induced by check node $\mathtt{c}$ is satisfied. For later use, we note that the calculation of ${{\rm{CTV}}_{\mathtt{c} \to \mathtt{v}}^{(l)}}$ can be represented compactly, as follows. Define the \textit{sumset} (or \textit{Minkowski sum}) \cite{Tao} operation between sets $\mathcal{S}_j$ ($j=1,2,...,J$) that contain GF($q$) elements:
\begin{equation}
\label{minkowski}
\sum\limits_{j = 1}^J {{{\cal S}_j}}  = \left\{ {\sum\limits_{j = 1}^J {{\mathcal{S}_j}} :{s_j} \in {\mathcal{S}_j}} \right\}.
\end{equation}
That is, the sumset results in a set containing all sums (using GF($q$) arithmetic) of elements taken from  $\mathcal{S}_j$. 
\begin{example}
Assume that $q=4$ and consider the sets $\left\{ {0,{\alpha ^0}} \right\}$ and $\left\{ {0,{\alpha ^1}} \right\}$. The sumset of these sets is $\left\{ {0,{\alpha ^0}} \right\} + \left\{ {0,{\alpha ^1}} \right\} = \left\{ {0 + 0,0 + {\alpha ^1},{\alpha ^0} + 0,{\alpha ^0} + {\alpha ^1}} \right\} = \left\{ {0,{\alpha ^0},{\alpha ^1},{\alpha ^2}} \right\}$, i.e., all the field elements. On the other hand, if both sets are $\left\{ {0,{\alpha ^0}} \right\}$, then the sumset is $\left\{ {0,{\alpha ^0}} \right\} + {\rm{ }}\left\{ {0,{\alpha ^0}} \right\}\ = \left\{ {0,{\alpha ^0}} \right\}$. 
\end{example}
For each pair of check node $\mathtt{c}$ and variable node $\mathtt{v}$, define the following sets for each $\mathtt{v}' \in \left\{ {\mathcal{N}\left( \mathtt{c} \right)\backslash \mathtt{v}} \right\}$ :
\begin{equation}
\label{X_i_def}
\mathcal{B}_{\mathtt{v}'}^{(l)} = \left\{ { - {{{h_{\mathtt{c},\mathtt{v}'}} \cdot {g'}}}{{}}:{g'} \in {\rm VTC}_{\mathtt{v}' \to \mathtt{c}}^{\left( {l - 1} \right)} } \right\}, \hspace{10pt} l \ge 1
\end{equation}
which are the VTC message sets sent from the variable nodes adjacent to $\mathtt{c}$ except $\mathtt{v}$ to check node $\mathtt{c}$ at iteration $l-1$, multiplied by the additive inverses of the edge labels. We can then represent the calculation of CTV messages in a compact manner:
\begin{equation}
\label{CTV}
{\rm CTV}_{\mathtt{c} \to \mathtt{v}}^{\left( l \right)} =\frac{1}{h_{\mathtt{c},\mathtt{v}}} \sum\limits_{\mathtt{v}' \in \left\{ {\mathcal{N}\left( \mathtt{c} \right)\backslash \mathtt{v}} \right\}} {\mathcal{B}_{\mathtt{v}'}^{\left( l \right)}}, \hspace{10pt} l \ge 1.
\end{equation}
In words, ${\rm CTV}_{\mathtt{c} \to \mathtt{v}}^{\left( l \right)}$ is the set of  possible values of $\mathtt{v}$ given the incoming VTC messages from the variable nodes in $\left\{ {\mathcal{N}\left( \mathtt{c} \right)\backslash \mathtt{v}} \right\}$, where an example is given in Figure \ref{fig_CTV}. Note that a CTV message can be of cardinality between $1$ and $q$. We now move to calculate the VTC messages, which are based on the CTV messages. The VTC messages are calculated as follows:
\begin{equation}
\label{VTC}
{\rm VTC}_{\mathtt{v} \to \mathtt{c}}^{\left( l \right)} = {\rm VTC}_{\mathtt{v}}^{\left( 0 \right)}\bigcap {\left( {\bigcap\limits_{\mathtt{c}' \in \left\{ {{\cal N}\left( \mathtt{v} \right)\backslash \mathtt{c}} \right\}} {{\rm CTV}_{\mathtt{c}' \to \mathtt{v}}^{\left( l \right)}} } \right)}, \hspace{10pt} l \ge 1.
\end{equation}
That is, the VTC message ${\rm VTC}_{\mathtt{v} \to \mathtt{c}}^{\left( l \right)}$ is the set of symbols containing the \textit{intersection} of the channel information and the incoming CTV messages to variable node $\mathtt{v}$, where an example is given in Figure \ref{fig_VTC}. A VTC message cardinality can be at most $M$, as the channel information cardinality is at most $M$. A decoding failure occurs if at the end of the decoding process there is a VTC message containing more than one symbol.

The decoding process described above reduces to the known iterative decoder proposed for the BEC/QEC \cite{mct} when $M=q$. In this case, the passed messages can be either the set containing all $q$ symbols (full erasure) or a set containing the correct symbol only. This greatly simplifies the asymptotic iterative-decoding performance analysis of LDPC codes when used over the BEC \cite{mct}. However, apart from erasure/non-erasure messages in the BEC case, there are many other possible message sets in the QPEC decoding process, making the analysis prohibitively complex as $q$ increases. In the following section, we discuss our approach for low-complexity approximate asymptotic decoding performance analysis, which is later shown to capture the exact behaviour quite well.

\section{Density Evolution Analysis}
\label{de_analysis}

Density evolution analysis of decoding performance is carried out by tracking the asymptotic (in the codeword length) probability of decoding failure at each iteration based on the probabilities of the passed messages \cite{RU, Bennatan, Richardson1}. In this section, we use this method for asymptotic performance evaluation of the decoder described in Section \ref{message_passing}. As customary, we assume a randomly constructed Tanner graph with a degree-distribution pair $\lambda$ and $\rho$, and a random i.i.d. selection of edge labels distributed uniformly on the non-zero elements of GF($q$). In addition, a sufficiently large codeword length is assumed, such that incoming messages to each node at each iteration of the decoding process are statistically independent with high probability (known as the \textit{independence assumption}) \cite{RU}. We start with deriving the exact QPEC density-evolution equations, and then move to propose approximate density evolution analysis due to complexity reasons.

Let us denote by $\mathcal{S}_t$, $t=1,2,...,2^q-1$, the non-empty subsets of the input alphabet $\mathcal{X}$ (of $q$ symbols), ordered by cardinality and in lexicographical order. These subsets may be passed throughout the decoding process as either VTC or CTV messages (see Section \ref{message_passing}). Denote by $z_t^{(l)}$ the probability that a VTC message at iteration $l$ is $\mathcal{S}_t$. Similarly, denote by $w_t^{(l)}$ the probability that a CTV message at iteration $l$ is $\mathcal{S}_t$. $\mathcal{I}_{{\bar d}}$ (resp. $\mathcal{J}_{{\bar d}}$) will denote an \textit{ordered list} containing ${{\bar d}}=d-1$ indices taken (with possible repetitions) from the set of message indices $\left\{ {1,2,...,{2^{q}-1}} \right\}$, representing VTC (resp. CTV) messages to a degree-$d$ check (resp. variable) node. Enumerating the edges connected to a check (variable) node $1$ to $\bar{d}$, an element in $\mathcal{I}_{{\bar d}}$ (resp. $\mathcal{J}_{{\bar d}}$) is the index of the message on the corresponding edge. For example, there are $15^2=225$ ordered lists $\mathcal{I}_2$ for a degree-$3$ check node when $q=4$: $(1,1), \left( {1,2} \right),\left( {2,1} \right), (2,2),...,(15,15)$, where the elements of $\mathcal{I}_2$ are the first and second incoming message indices. ${{\chi_t}\left( {{\mathcal{I}_{\bar d}}} \right)}$ will denote the probability that the VTC messages indexed in $\mathcal{I}_{\bar d}$ lead to the CTV message $\mathcal{S}_t$. Similarly, $\eta_t \left( \mathcal{J}_{\bar d} \right)$ will denote the probability that the CTV messages indexed in $\mathcal{J}_{\bar d}$ lead to the VTC message $\mathcal{S}_t$. The distributions $\chi_t$ and $\eta_t$ are obtained with respect to the uniform edge labels and the channel information, as demonstrated in the following example.
\begin{example}
\label{calc_chi_eta}
Assume a degree-$3$ check node and that $q=4$. Consider $\mathcal{I}_2=(5,5)$ and recall that according to our convention, $\mathcal{S}_5 = \left\{ {0,{\alpha ^0}} \right\}$. To calculate $\chi_t (\mathcal{I}_2)$, we find all possible outcomes of the sumset $\left( {{h_1}/{h_3}} \right) \cdot \left\{ {0,{\alpha ^0}} \right\} + \left( {{h_2}/{h_3}} \right) \cdot \left\{ {0,{\alpha ^0}} \right\}$ where $h_1, h_2$ and $h_3$ are i.i.d. random variables uniformly distributed on $\left\{ {{\alpha ^0},{\alpha ^1},{\alpha ^2}} \right\}$, representing the edge labels. If $h_1=h_2=h_3$, the sumset is $\left\{ {0,{\alpha ^0}} \right\} + \left\{ {0,{\alpha ^0}} \right\}\ = \left\{ {0,{\alpha ^0}} \right\}=\mathcal{S}_{5}$. On the other hand, if the edge labels are not the same, the sumset is $\left\{ {0,\alpha^0,{\alpha ^1},\alpha^2} \right\}={\mathcal{S}_{15}}$. Therefore, the non-zero $\chi_t$ values are ${\chi _5} = 1/9$ and ${\chi_{15}} = 8/9$ in this case. Now consider a degree-$3$ variable node, where $\mathcal{J}_2 = (6,6)$ and the channel information sets are $\mathcal{S}_5 = \left\{ {0,{\alpha ^0}} \right\}$, $\mathcal{S}_6 = \left\{ {0,{\alpha ^1}} \right\}$ and $\mathcal{S}_7 = \left\{ {0,{\alpha ^2}} \right\}$ (i.e, $M=2$), each with probability $1/3$. If the channel information is $\mathcal{S}_5$ or $\mathcal{S}_7$, then the intersection between the messages indexed in $\mathcal{J}_2$ and the channel information is $\mathcal{S}_1=\left\{ {0} \right\}$. If the channel information is $\mathcal{S}_6$, the intersection results in $\mathcal{S}_6$. Therefore, we get that for $\mathcal{J}_2=(2,2)$, the non-zero $\eta_t$ values are $\eta_1 = 2/3$ and $\eta_6 = 1/3$.
\end{example}

As GF($q$) LDPC codes are linear codes, the probability of a given codeword symbol taken from the codebook is $1/q$. This means that a variable node contains a certain set composed of one symbol (i.e., non-erasure) with probability $(1-\varepsilon)/q$. To incorporate this probability in the density-evolution equations, we define the indicator ${\theta _t}$, which equals $1$ if $\left| {{\mathcal{S}_t}} \right| = 1$ and $0$ otherwise. Equipped with the notations above, we get the following compact representation of the QPEC density-evolution equations:
\begin{equation}
\label{exact_de1}
w_t^{\left( l \right)} = \sum\limits_{i = 2}^{{d_c}} {{\rho _i}} \sum\limits_{{{\cal I}_{i - 1}}} {\left( {\prod\limits_{j \in {{\cal I}_{i - 1}}} {z_j^{\left( {l - 1} \right)}} } \right) \cdot {\chi _t}\left( {{{\cal I}_{i - 1}}} \right)},
\end{equation}
\begin{equation}
\label{exact_de2}
z_t^{\left( l \right)} = \varepsilon \sum\limits_{i = 2}^{{d_v}} {{\lambda _i}\sum\limits_{{{\cal J}_{i - 1}}} {\left( {\prod\limits_{j \in {{\cal J}_{i - 1}}} {w_j^{\left( l \right)}} } \right) \cdot {\eta _t}\left( {{{\cal J}_{i - 1}}} \right)} }  + \frac{(1 - \varepsilon )}{q} \cdot \theta_t.
\end{equation}
The summation over $\mathcal{I}_{i-1}$ (or $\mathcal{J}_{i-1}$) is understood over all ordered lists containing $i-1$ elements (where $i$ is the node degree) taken from the set of indices $\left\{ {1,2,...,{2^{q}-1}} \right\}$. A decoding failure occurs when a variable node is not resolved, i.e., when it contains a set with more than one symbol:
\begin{equation}
\label{Perror_exact}
p_e^{(l)} = \sum\limits_{t:\left| {{\mathcal{S}_t}} \right| > 1} {z_t^{\left( l \right)}}  = 1 - \sum\limits_{t:\left| {{\mathcal{S}_t}} \right| = 1} {z_t^{\left( l \right)}}.
\end{equation}
Note that for $M=q$, only $z_1, w_1, z_{2^q-1}$ and $w_{2^q-1}$ (i.e., probabilities of full-erasure/non-erasure sets) might be positive. In this case, these probabilities can be represented solely by $z_{2^q-1}$, as the distributions $\chi_t$ and $\eta_t$ degenerate due to the simple BEC/QEC decoding rules. Equations \eqref{exact_de1}-\eqref{exact_de2} can be then readily simplified to obtain the BEC/QEC (one-dimensional) density-evolution equations \cite{mct}.

Calculating ${{\chi_t}\left( {{\mathcal{I}_{\bar d}}} \right)}$ and $\eta_t \left( \mathcal{J}_{\bar d} \right)$ in Equations \eqref{exact_de1}-\eqref{exact_de2} might be prohibitive in practice, as the number of subsets increases exponentially with $q$. To get an estimate of the complexity, consider basic check and variable nodes of degree $3$. Given two incoming message sets to the check, calculating the distribution ${{\chi _t}}$ requires ${\left( {q - 1} \right)^3}$ realizations of edge labels. Because there are $\mathcal{O}\left( {{2^q}} \right)$ input-set pairs, we get $\mathcal{O}\left( {{q^3} \cdot {2^{2q}}} \right)$ complexity for calculating $\chi_t$. In a similar manner, ${\cal O}\left( {{q \choose M} \cdot {2^{2q}}} \right)$ operations are required for calculating ${{\eta _t}}$ (the first factor now being the number of possible channel-information sets) for a degree-$3$ variable node. As an example, about $10^{13}$ operations are required for the calculation of $\chi_t$ when $q=16$, growing to the order of $10^{23}$ when $q=32$. In addition to prohibitive complexity, the exhaustive calculation of ${\chi_t}$ and ${\eta_t}$ as demonstrated in Example \ref{calc_chi_eta} provides no insights on their behaviour. Moreover, ${\chi _t}$ requires the explicit use of GF($q$) arithmetic, making its analysis difficult. These reasons motivate us to propose a more efficient way for estimating the QPEC decoding performance, which we discuss in the following subsection.



\subsection{Cardinality-based approximated density-evolution equations}
\label{de_equations}

To overcome the difficulties in evaluating Equations \eqref{exact_de1}-\eqref{exact_de2}, we propose to track the probability distribution of the VTC/CTV message set \textit{cardinalities}. In our approach, we approximate messages of the same cardinality passed in the decoding process as being equiprobable. The intuition behind this approximation comes from the randomness of the edge labels and the channel output that "smoothen" most of the non-uniformity that may occur due to the algebraic structure of GF($q$). In particular, as the node degrees and the field order grow, the incidence probability of equal-cardinality sets becomes increasingly uniform. The reason is that the entropy of each sum in the sumset performed at check nodes increases with the degree. In addition, the number of sums within the sumset increases with $q$, increasing the entropy of the sumset result. The approximation was verified empirically as well, where we show in Section \ref{subsection:comparison} that performance analyzed with this assumption gives a very good approximation of the true decoding performance.

To distinguish between message sets and their cardinalities, we use the notation $\mathcal{M}_{\bar d}$ to denote an ordered list of $\bar d=d-1$ elements taken from $\left\{ {1,2,...,q} \right\}$, understood as possible incoming message-set cardinalities to a degree $d$ check node. $W_m^{(l)}$ (resp. $Z_m^{(l)}$) will denote the probability that a CTV (VTC) message at iteration $l$ is of cardinality $m=1,2,...,q$. $P_m(\mathcal{M}_{\bar d})$ (resp. $Q_m(\mathcal{M}_{\bar d})$) will denote the probability that the message-set cardinalities in $\mathcal{M}_{\bar d}$ lead to an outgoing CTV (VTC) message of cardinality $m$. Note that the distributions $P_m$ and $Q_m$ are obtained by summing the probabilities of $\chi_t$ and $\eta_t$ for all $t$ with $|\mathcal{S}_t|=m$, assuming uniform distribution on the input sets with cardinalities in $\mathcal{M}_{\bar d}$. Finally, under our approximation, the following equations are derived:
\begin{equation}
\label{cardinality_de1}
W_m^{\left( l \right)} \simeq \sum\limits_{i = 2}^{{d_c}} {{\rho _i}}  \cdot \sum\limits_{\mathcal{M}_{i - 1}} {\left( {\prod\limits_{m' \in \mathcal{M}_{i-1}} {Z_{m'}^{\left( {l - 1} \right)}} } \right) \cdot {P_m}\left( {{\mathcal{M}_{i - 1}}} \right)},
\end{equation}
\begin{align}
\label{cardinality_de2}
Z_m^{\left( l \right)} \simeq & { \varepsilon \cdot \sum\limits_{i = 2}^{{d_v}} {{\lambda _i}}  \cdot \sum\limits_{\mathcal{M}_{i-1}} {\left( {\prod\limits_{m' \in \mathcal{M}_{i-1}} {W_{m'}^{\left( l \right)}} } \right) \cdot {Q_m}\left( {{\mathcal{M}_{i - 1}}} \right)}} 
\\\nonumber
&+ (1-\varepsilon)  \cdot \delta \left[ {m - 1} \right],
\end{align}
where $\delta \left[ {m } \right]$ is the discrete Dirac delta function. The summation over $\mathcal{M}_{i-1}$ is understood over the ordered lists of $i-1$ elements taken from the set of possible incoming message-set cardinalities. This set is $\left\{ {1,2,...,M} \right\}$ for incoming VTC and $\left\{ {1,2,...,q} \right\}$ for incoming CTV message-set cardinalities. The initial conditions are $Z_1^{\left( 0 \right)} = 1 - \varepsilon ,Z_M^{\left( 0 \right)} = \varepsilon $ and $Z_m^{\left( 0 \right)} = 0$ for $m \ne 1, M$. The asymptotic probability of decoding failure at iteration $l$ is the probability of a VTC message-set cardinality larger than $1$ at iteration $l$:
\begin{equation}
\label{Perror_cardinality}
p_e^{(l)} = \sum\limits_{m = 2}^q {Z_m^{\left( l \right)}}  = 1 - Z_1^{\left( l \right)}.
\end{equation}

We note here that in our experiments the probability of decoding failure calculated using \eqref{Perror_cardinality} is virtually the same as \eqref{Perror_exact} even for small $q$ and check-node degree values, such that the cardinality-based equations can be safely used for QPEC performance evaluation. However, though we moved from $\mathcal{O}(2^q)$ possible message sets to $q$ possible message-set cardinalities, we still need efficient ways to calculate $P_m$ and $Q_m$. A straightforward calculation enumerates $\chi_t$ and $\eta_t$ for $\mathcal{O}(2^{2q})$ realizations of message set pairs, which does not quite solve the complexity problem. Thus we devote the remainder of this section and the next section to efficient calculations, bounding, and approximations for $P_m$ and $Q_m$. We begin with providing in Section \ref{PDVTC} an exact closed-form expression for $Q_m$. In Section \ref{pm_approx} we show that finding a closed-form expression for $P_m$ is hard. Therefore, we propose computationally efficient bounds and approximation models for $P_m$. We later use our models and bounds to determine the QPEC decoding threshold and to design good LDPC codes.

\subsection{Formula for $Q_m$}
\label{PDVTC}

$Q_m(\mathcal{M}_{\bar d})$ is the probability of an intersection of cardinality $m$ between CTV messages with cardinalities taken from $\mathcal{M}_{\bar d}$ and a channel information set of cardinality $M$, where message sets of the same cardinality are equiprobable. Define $\mathcal{M}_{d}$ to contain the cardinalities in $\mathcal{M}_{\bar d}$ together with the channel information set cardinality $M$ and $\mu$ to be the smallest cardinality in $\mathcal{M}_{d}$, i.e. $\mu \buildrel \Delta \over = {\min _{m' \in {\mathcal{M}_{d}}}}m'$. In the following, we find the number of ways to realize the sets in $\mathcal{M}_{d}$ such that their intersection is of cardinality $m$, and later take into account the presence of the correct symbol in each set. We begin with the following lemma.

\begin{lemma}
\label{Km_calculation}
\emph{(Number of ways to get an intersection of cardinality $m$)}
Consider $d$ message sets whose cardinalities are in $\mathcal{M}_{d}$. The number of ways to realize the sets such that their intersection is of cardinality $m$ ($m=0,1,...,\mu)$ is:
\begin{equation}
\label{Im_def}
K_m (\mathcal{M}_d; q) = \sum\limits_{s = 0}^{\mu  - m} {{{\left( { - 1} \right)}^s} \cdot } {\upsilon _{m + s}} \cdot {m+s \choose m},
\end{equation}
where
\begin{equation}
{\upsilon _{m + s}} \buildrel \Delta \over = {q \choose m+s} \cdot \prod\limits_{m' \in \mathcal{M}_d} {q-(m+s) \choose m'-(m+s)}.
\label{ni}
\end{equation}
\end{lemma}

\begin{proof}
Consider a fixed subset of $\mu$ elements taken from a set of $q$ elements. The number of ways to choose $d$ subsets with cardinalities in $\mathcal{M}_d$ such that they all contain the subset of $\mu$ elements is $\prod\limits_{m' \in \mathcal{M}_d} {q-\mu \choose m'-\mu}$, as we are free to choose only $m' - \mu $ elements for each subset of cardinality $m'$. Taking into account the number of ways to choose a subset of $\mu$ elements, which is ${q \choose \mu}$, we have 
\begin{equation}
K_{\mu} = {q \choose \mu} \cdot \prod\limits_{m' \in \mathcal{M}_d} {q-\mu \choose m'-\mu}
 = \upsilon_\mu 
\end{equation}
ways to choose the subsets such that their intersection is of cardinality $\mu$. To find $K_m$ for $m=\mu-1$, we proceed as follows. The number of ways to choose the subsets such that they contain a fixed subset of $\mu -1$ elements is $\prod\limits_{m' \in \mathcal{M}_d}{q-(\mu-1) \choose m'-(\mu-1)}$. However, the subsets may also contain a subset of cardinality $\mu$ such that the fixed subset of cardinality $\mu-1$ is its subset, resulting in overcounting. Since there are ${\mu \choose \mu-1}=\mu$ sets of cardinality $\mu-1$ contained in a set of cardinality $\mu$, we correct for overcounting as follows:
\begin{align}
{K_{\mu  - 1}} = & {q \choose \mu-1} \cdot \prod\limits_{m' \in \mathcal{M}_d} {q-(\mu-1) \choose m'-(\mu-1)}  - \mu  \cdot {\upsilon _\mu }
\\\nonumber = &{{\upsilon _{\mu  - 1}} - \mu  \cdot {\upsilon _\mu }}.
\end{align}
Moving to $\mu -2$, we first count sets of cardinality $\mu-2$ with $\upsilon_{\mu-2}$ and  then subtract ${\mu-1 \choose \mu-2} \cdot {\upsilon _{\mu  - 1}}$ sets to account for sets of cardinality $\mu-1$. However, we now over-correct some sets of cardinality $\mu$. We account for that by considering the ${\mu \choose \mu-2}$ sets of cardinality $\mu-2$ contained in a set of cardinality $\mu$ to obtain:
\begin{equation}
{K_{\mu  - 2}} = \upsilon_{\mu-2} - {\mu-1 \choose \mu-2}  \cdot \upsilon_{\mu-1} + {\mu \choose \mu-2} \cdot \upsilon_{\mu}.
\end{equation}
Continuing in the same fashion (essentially, we use the inclusion-exclusion principle), we get:
\begin{equation}
{K_{\mu  - t}} = \sum\limits_{i = 0}^t {{{\left( { - 1} \right)}^i} \cdot } {\upsilon _{\mu  - t + i}} \cdot {\mu-t+i \choose \mu-t},
\label{Kmu}
\end{equation}
for $t=0,1,...,\mu$. Index shifting leads to the desired result.
\qed
\end{proof}
We are now ready to provide a formula for $Q_m$. Lets us denote by $\mathcal{M}_d -1$ the ordered list obtained by subtracting $1$ from each number (set cardinality) in $\mathcal{M}_d$.

\begin{theorem}
\emph{(Formula for $Q_m$)}
\begin{equation}
\label{Qm_formula}
{Q_m}\left( {{\mathcal{M}_{\bar d}}} \right) = \left\{ {\begin{array}{*{20}{c}}
{\frac{{{K_{m - 1}}(\mathcal{M}_d -1;q-1)}}{{\prod\limits_{m' \in \mathcal{M}_d} {q-1 \choose m'-1} }},}&{{\rm{if \hspace{5pt}}}\mu   > 1}\\
{\delta \left[ {m - 1} \right],}&{{\rm{otherwise.}}}
\end{array}} \right.
\end{equation}
\end{theorem}
\begin{proof}
We use $K_{m-1}$, $\mathcal{M}_d -1$ and $q-1$ as we can choose effectively $m'-1$ elements for each subset of cardinality $m'$, as the correct symbol appears in the subsets. We then normalize by the number of subsets with cardinalities ${m' - 1}$ taken from a set of $q-1$ elements to obtain a probability distribution. Note that when $\mu=1$ the intersection is necessarily of cardinality $1$, such that that $Q_1 = 1$. \qed
\end{proof}

\section{Bounds and approximations for $P_m$}
\label{pm_approx}

$P_m(\mathcal{M}_{\bar d})$ in Equation \eqref{cardinality_de1} is the probability that the sumset of the sets with cardinalities in $\mathcal{M}_{\bar d}$ is of cardinality $m$, where sets of the same cardinality are equiprobable and the edge labels are uniformly distributed. Considering all possible realizations of the messages becomes intractable as the field size or the node degree increase. The major reason for the difficulty in calculating $P_m$ (unlike $Q_m$) is that it involves GF($q$) arithmetic. Thus, finding a closed-form expression for $P_m$ is hard, see e.g. the discussion on sumsets in \cite{Tao}, \cite{Eliahou199812}, \cite{croot}. Because of that, we seek instead efficient bounds and approximations for $P_m$. Let $\mathcal{I}_{\bar d}$ contain indices of arbitrary message sets whose cardinalities are in $\mathcal{M}_{\bar d}$. Denote $\kappa \buildrel \Delta \over = {\max _{m' \in {\mathcal{M}_{\bar d}}}}m'$ as the maximal number (set cardinality) in $\mathcal{M}_{\bar d}$. In addition, denote $N \buildrel \Delta \over = \prod\limits_{m' \in {\mathcal{M}_{\bar d}}} {m'} $ as the number of sums in the calculation of $\sum\limits_{j \in {{\cal I}_{\bar d}}} {{{\cal S}_j}}$.

\begin{example}
Assume that $q=4$ and $\bar{d}=2$. If $\mathcal{M}_{\bar d} = \left\{ {2,3} \right\}$, then the first element in $\mathcal{I}_{\bar d}$ can be between $5$ and $10$, and the second element can be between $11$ and $14$.
\end{example}

\subsection{Upper and lower bounds on $P_m$ using additive combinatorics}
\label{pm_bounds}

In this subsection we derive bounds on the cardinality of the sumset $\sum\limits_{j \in {{\cal I}_{\bar d}}} {{{\cal S}_j}}$. These bounds will be a function of the message-set cardinalities $\mathcal{M}_{\bar d}$, such that they are universal for all realizations of sets adhering to the cardinalities in $\mathcal{M}_{\bar d}$. We begin with simple lower and upper bounds.
\begin{lemma}
\label{simple_bounds}
\emph{(Simple bounds on a sumset cardinality \cite{Tao})}
\begin{equation}
\kappa \le \left| {\sum\limits_{j \in {\mathcal{I}_{\bar d}}} {{\mathcal{S}_j}} } \right| \le \min \left( q,N \right).
\end{equation}
\end{lemma}

The following lemma provides a sufficient condition for attaining the maximal sumset cardinality $q$.

\begin{lemma}
\label{sct}
\emph{(Sufficient condition for the sumset of cardinality $q$ \cite{Tao})} If there are $m,m' \in \mathcal{M}_{\bar d}$ (where $m$ and $m'$ are taken from two different positions in $\mathcal{M}_{\bar d}$) such that $m+m' > q$, then ${\left| {\sum\limits_{j \in \mathcal{I}_{\bar d}} {{\mathcal{S}_j}} } \right| = q}$.
\end{lemma}
For later use, we say that the {\it q-condition} holds if the condition of Lemma \ref{sct} is satisfied. Note that this condition can be satisfied only if $M>q/2$. We now proceed to obtain improved lower bounds on the sumset cardinality, using the following two theorems.

\begin{theorem}
\label{cauchy_davenport}
\emph{(Cauchy-Davenport Theorem \cite{Tao})} Consider the finite field GF($p$), $p$ prime. Let $\mathcal{S}_a$ and $\mathcal{S}_b$ be two non-empty subsets of GF($p$). Then:
\begin{equation}
\label{CDI}
\left| {\mathcal{S}_a + \mathcal{S}_b} \right| \ge \min \left( {p,\left| \mathcal{S}_a \right| + \left| \mathcal{S}_b \right| - 1} \right).
\end{equation}
\end{theorem}
The following theorem by K\'{a}rolyi provides an extension of the Cauchy-Davenport theorem to finite groups.

\begin{theorem}
\label{karolyi}
\emph{(K\'{a}rolyi's theorem for finite groups \cite{Karolyi})} Let $\mathcal{S}_a$ and $\mathcal{S}_b$ be two non-empty subsets of a finite group $G$. Denote by $p\left( G \right)$ the smallest prime factor of $\left| G \right|$. Then:
\begin{equation}
\left| \mathcal{S}_a + \mathcal{S}_b \right| \ge \min \left( {p\left( G \right),\left| \mathcal{S}_a \right| + \left| \mathcal{S}_b \right| - 1} \right).
\end{equation}
\end{theorem}
This theorem can be used for extending the inequality \eqref{CDI} to extension fields, as we have in the following theorem.

\begin{theorem}
\label{cor_bounds}
\emph{(Improved sumset cardinality bounds)} Denote by $p$ the prime factor of $q$. Then:
\begin{align}
\label{sumset_improved_bounds}
\max \left( {\kappa ,\min \left( {p,\sum\limits_{m' \in {\mathcal{M}_{\bar d}}} {m' - {\bar d} + 1} } \right)} \right) & \le \left| {\sum\limits_{j \in {{\cal I}_{\bar d}}} {{{\cal S}_j}} } \right| 
\\\nonumber
& \le \min \left( {q,N} \right).
\end{align}
\end{theorem}
\begin{proof}
This theorem is proved by Lemma \ref{simple_bounds} and Theorem \ref{karolyi}, followed by induction on the number of subsets (see e.g. \cite{Grynkiewicz} for the proof technique when $q$ is prime).
\qed
\end{proof}
The bounds of Theorem \ref{cor_bounds} are sharp (i.e., there exist subsets $\mathcal{S}_j$ with cardinalities in $\mathcal{M}_{\bar d}$ such that the bounds are attained) \cite{Tao}. We will denote by $B_L$ and $B_U$ the lower and upper bounds of inequality \eqref{sumset_improved_bounds}, respectively. We use these bounds to derive two bounding distributions $P_m^{\left( {{\text{max}}} \right)}$ and $P_m^{\left( {{\text{min}}} \right)}$: the former to bound the output set cardinalities from above, and the latter from below. To get $P_m^{\left( {{\text{max}}} \right)}$, the sumset is assumed as of cardinality $B_U$ with probability $1$, unless the $q$-condition is satisfied. 
\begin{equation}
\label{PmMax}
P_m^{\left( {{\text{max}}} \right)} = \left\{ {\begin{array}{*{20}{c}}
{\delta \left[ {m - q} \right],}&{{\text{if the } q\text{-condition holds}}}\\
{\delta \left[ {m - {B_U}} \right],}&{{\text{otherwise}}}.
\end{array}} \right.
\end{equation}
In a similar manner, $P_m^{\left( {\text{min}} \right)}$ is calculated using the lower bound $B_L$ on the sumset cardinality:
\begin{equation}
\label{PmMin}
P_m^{\left( {\text{min}}\right)} = \left\{ {\begin{array}{*{20}{c}}
{\delta \left[ {m - q} \right],}&{{\text{if the } q\text{-condition holds}}}\\
{\delta \left[ {m - {B_L}} \right],}&{{\text{otherwise}}}.
\end{array}} \right.
\end{equation}
The importance of $P_m^{\left( {\text{max}} \right)}$ resp. $P_m^{\left( {\text{min}} \right)}$ is that using them in the density evolution iteration in place of the true $P_m$ gives a lower resp. upper bound on the asymptotic probability of decoding failure \eqref{Perror_cardinality} calculated using the cardinality-based density-evolution equations.
 
Going beyond the bounds above to a potentially tighter characterization of $P_m$, in the remainder of the section we propose two low-complexity approximation models for $P_m$. We begin with a simple balls-and-bins model, and later refine it with a tighter model. Finally, we compare the bounds above with the proposed approximation models.

\subsection{The balls-and-bins model}
\label{Balls_and_bins}

The major difficulty in calculating $P_m$ exactly is its dependence on the structure of the finite-field arithmetic. Going around this difficulty, we propose a pure-probabilistic approximation of $P_m$ using the \textit{balls-and-bins model} \cite{Mitz}. In this model, balls are placed independently and uniformly at random to bins, where we are usually interested in the distribution of the number of non-empty bins once all the balls were placed. Motivated by the randomness induced by the random edge labels, we propose to consider the $N$ sums in the calculation of the sumset as the balls, and the $q$ elements of GF($q$) as the bins. This way, $P_m$ is modeled as the probability of $m$ non-empty bins after the $N$ balls were placed. As a consequence, the use of GF($q$) arithmetic is not required when the balls-and-bins model is used.

The balls-and-bins model is an \textit{absorbing} Markov process with $q+1$ possible \textit{states}, with state $m$ ($m=0,1,...,q$) corresponding to $m$ non-empty bins out of $q$. The absorbing state is $q$, as once $q$ bins are non-empty the number of non-empty bins cannot change. The $\left( {q + 1} \right) \times \left( {q + 1} \right)$ Markov matrix describing this process takes a simple form, since we can either stay at state $m$ or move to state $m+1$. Denoting the Markov matrix as ${{\bf{\Gamma }}_{{\rm{balls}}}}$, its entries are:
\begin{equation}
\label{gamma_balls}
{\left( {{{\bf{\Gamma }}_{{\rm{balls}}}}} \right)_{m,m}} = \frac{m}{q},{\left( {{{\bf{\Gamma }}_{{\rm{balls}}}}} \right)_{m,m + 1}} = 1 - \frac{m}{q},
\end{equation}
where the remaining entries are zeros. That is, if the current state is $m$, then a ball is placed in a one of the $m$ non-empty bins with probability $m/q$, and is placed in a different bin with probability $1-m/q$. Let us denote by ${{\boldsymbol{g}}^{\left( N \right)}} = \left( {g_0^{\left( N \right)},g_1^{\left( N \right)},...,g_{q}^{\left( N \right)}} \right)$ the probability distribution on the states defined by ${{{\bf{\Gamma }}_{{\text{balls}}}}}$, where $g_m^{\left( N \right)}$ is the probability of state $m$ after the $N$ balls were placed. According to the Markov property, ${{\boldsymbol{g}}^{\left( N \right)}} = {\boldsymbol{g}^{\left( 0 \right)}} \cdot {\bf{\Gamma }}_{{\rm{balls}}}^N$ (where ${\bf{\Gamma }}_{{\rm{balls}}}^N$ is ${{\bf{\Gamma }}_{{\rm{balls}}}}$ raised to power $N$). As ${\boldsymbol{g}}^{(0)}= \left( {1,0,...,0} \right)$ (i.e., the bins are empty at the beginning), ${\boldsymbol{g}}^{\left( N \right)}$ is simply the first row of ${{\bf{\Gamma }}_{\text {balls}}^{N}}$. Finally, using the $q$-condition (Lemma \ref{sct}) and the lower bound $B_L$ (see Section \ref{pm_bounds}), we define the following approximation model for $P_m$:
\begin{equation}
\label{pm_balls}
P_m^{\left( {{\text{balls}}} \right)} = \left\{ {\begin{array}{*{20}{c}}
{0,}&{{\text{if }}m < {B_L}}\\
{\delta \left[ {m - q} \right],}&{{\text{if the }}q{\text{-condition holds}}}\\
{\frac{{g_m^{\left( N \right)}}}{{\sum\limits_{m' = {B_L}}^q {g_{m'}^{\left( N \right)}} }},}&{{\text{otherwise}}}.
\end{array}} \right.
\end{equation}

The expected number of balls required to get into the absorbing state $q$ (when starting at state $0$) is $q\ln q + q\gamma$ (up to $\mathcal{O}\left( {1/\left( {2q} \right)} \right)$ terms), where $\gamma \approx 0.577$ is the {\it Euler-Mascheroni constant} \cite{Mitz}. That is, all the bins will be non-empty on average when $N \approx q\ln q + q \cdot \gamma$, which can be thought as the probabilistic extension of the $q$-condition to the balls-and-bins model. For such $N$ values, the sumset cardinality approximated using the balls-and-bins model is expected to be $q$ and ${{\boldsymbol{g}}^{\left( N \right)}}$ degenerates to the absorbing distribution $(0,\ldots,0,1)$. Therefore, ${{\bf{\Gamma }}_{\text {balls}}^{N}}$ should be calculated in practice for values of $N$ up to approximately $q\ln q + q\gamma$, even for high-degree check nodes.

\subsection{The union model}

In the previous sub-section, we modeled $P_m$ using the balls-and-bins model, where each ball (which corresponds to an element obtained by a sum within the sumset) is independent of the other balls. In this part, we improve this approximation by exploiting an important property of the $N$ sums within the sumset: they can be divided into $N/\kappa$ sets of $\kappa$ (distinct) elements. This is proved by viewing the sums as generated by one element from the maximal-cardinality subset (of cardinality $\kappa$) and elements from the remaining subsets. This observation leads us to suggest a refined version of the balls-and-bins model, which we term as the \textit{union model}. In this model, the probability of a sumset of cardinality $m$ is modeled as the probability that the union of $N/\kappa$ random sets with cardinality $\kappa$ each results in a set of cardinality $m$. In view as balls-and-bins, it is the probability of $m$ non-empty bins after $\kappa$ groups of $N/\kappa$ balls are placed into the $q$ bins, where the balls in each group are placed uniformly at random into $\kappa$ \textit{distinct} bins.

Let us denote by ${{\boldsymbol{u}}^{\left( N/\kappa \right)}} = \left( {u_0^{\left( N/\kappa \right)},u_1^{\left(N/\kappa \right)},...,u_q^{\left( N/\kappa \right)}} \right)$ the probability distribution on the $q+1$ states after $N/\kappa$ groups of balls were placed into the bins. That is, $u_m^{\left( N/\kappa \right)}$ is the probability of state $m$ after $N/\kappa$ groups of $\kappa$ balls each were placed in the bins according to the union model. The transition probability ${P_{{\text{union}}}}\left( {m \to m'} \right)$ from state $m$ to state $m'$ is equivalent to the probability that the union of a random set of cardinality $m$ with a random set of cardinality $\kappa$ is of cardinality $m'$ (given that the set elements are taken from $q$ elements). To calculate this probability, denote by $\mathcal{A}$ a set of cardinality $m$ and by $\mathcal{B}$ a set of cardinality $\kappa$. We have:
\begin{align}
\label{ix1}
{P_{{\text{union}}}}\left( {m \to m'} \right) &= \Pr \left( {\left| {\mathcal{A} \cup \mathcal{B}} \right| = m'} \right) \\ \nonumber &= \Pr \left( {\left| {\mathcal{A} \cap \mathcal{B}} \right| = m+\kappa - m'} \right),
\end{align}
where we used the inclusion-exclusion principle. Thus, we can equivalently find the probability that the \textit{intersection} of the sets $\mathcal{A}$ and $\mathcal{B}$ is of cardinality $m+\kappa-m'$. Recall that ${K_{m + \kappa - m' }}\left( {m,\kappa } \right)$ (see Lemma \ref{Km_calculation}, $q$ is omitted for brevity) is the number of ways to obtain such an intersection cardinality. Dividing ${K_{m + \kappa - m' }}\left( {m,\kappa} \right)$ by the number of possible realizations of elements in the sets provides the desired probability ${P_{{\text{union}}}}\left( {m \to m'} \right)$. Therefore, the entries of the Markov matrix associated with the union model are:
\begin{equation}
\label{markov_matrix_union}
{\left( {{{\bf{\Gamma }}_{{\text{union}}}}} \right)_{m,m'}} = \frac{{{K_{m + \kappa -m' }}\left( {m,\kappa } \right)}}{{{q \choose m} \cdot {q \choose \kappa}}}.
\end{equation}
It is not hard to check that for $\kappa=1$, ${{{{\bf{\Gamma }}_{{\text{union}}}}} }$ reduces to ${{{{\bf{\Gamma }}_{{\text{balls}}}}}}$ (defined in \eqref{gamma_balls}). As before, the Markov property implies that ${{\boldsymbol{u}}^{\left( N/\kappa \right)}}$ is simply the first row of ${{\bf{\Gamma }}_{\text {union}}^{N/\kappa}}$. Finally, the following approximation for $P_m$ is based on the union model:
\begin{equation}
\label{Pm_union}
P_m^{\left( {{\text{union}}} \right)} = \left\{ {\begin{array}{*{20}{c}}
{0,}&{{\text{if }}m < {B_L}}\\
{\delta \left[ {m - q} \right],}&{{\text{if the }}q{\text{-condition holds}}}\\
{\frac{{u_m^{\left( {N/\kappa } \right)}}}{{\sum\limits_{m' = {B_L}}^q {u_{m'}^{\left( {N/\kappa } \right)}} }},}&{{\text{otherwise}}}
\end{array}} \right.
\end{equation}

\begin{figure}[t!]
\centering
\includegraphics[scale=0.55]{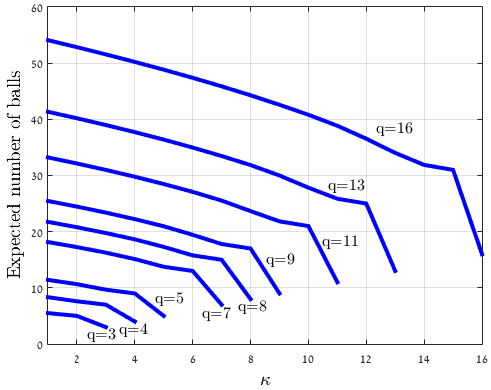}
\caption{The expected number of balls required for getting $q$ non-empty bins in the union model.}
\label{ex_kappa}
\end{figure}

To obtain the expected number of balls required to get into the absorbing state $q$, we use the {\it fundamental matrix} \cite{Grinstead} associated with an absorbing Markov chain. In our case, this matrix is ${\boldsymbol{\Phi} _{{\text{union}}}} = {\left( {{{\bf{I}}_q} - {{\bf{Q}}_{{\text{union}}}}} \right)^{ - 1}}$ where ${{{\mathbf{I}}_{q }}}$ is the identity matrix of dimensions $ q \times q$ and ${{{\bf{Q}}_{{\text{union}}}}}$ is the upper-left $ q \times q$ sub-matrix of ${{{\bf{\Gamma }}_{{\text{union}}}}}$. The expected number of groups of balls required to get into the absorbing state $q$ when starting with state $i_0$ is the $i{_0^{\text {th}}}$ entry of the vector ${{\mathbf{\Phi }}_\text{union}}{\mathbf{1}}$, where ${\mathbf{1}}$ is a column vector whose entries are all $1$ \cite{Grinstead} (to get the expected number of \textit{balls}, we multiply by $\kappa$). In Figure \ref{ex_kappa}, the expected number of balls required for getting from state $0$ to state $q$ is given as a function of $q$ and $\kappa$. As mentioned earlier (see Section \ref{Balls_and_bins}), this expected number can be used for extending the $q$-condition to the union model. Note that for $\kappa=1$ the union model is essentially the balls-and-bins model and we have ${{\bf{\Phi }}_{{\text{union}}}}{\bf{1}}\left( {i = 0} \right) \approx q\ln q + q\gamma$ as we saw earlier. 

\subsection{Comparison of the bounds and approximations}
\label{subsection:comparison}

In this part, we verify the tightness of our approximations by comparing the \textit{decoding threshold} \cite{RU} obtained from the exact and the approximate (cardinality-based) density-evolution equations. The QPEC decoding threshold (for a given degree-distribution pair), denoted $\varepsilon_{\rm th}$, is the maximal partial-erasure probability $\varepsilon$ such that the probability of decoding failure tends to zero. Its operational meaning is the robustness of the iterative decoder to partially-erased codeword symbols, i.e., the fraction of partially-erased codeword symbols that the decoder can tolerate. In Figure \ref{fig_results}, we plot the exact and approximate decoding threshold values (using the bounds and models for $P_m$) for several values of $q$ and $M$ for the regular $(3,6)$ LDPC code ensemble (of rate $1/2$). We note that the exact threshold for $q=16$ is not provided in Figure \ref{fig_results} due to complexity reasons. When $M=q$, the QPEC density-evolution equations are equivalent to the BEC/QEC density-evolution equation (see Section \ref{de_analysis}). In this case, all the models and bounds give the exact threshold, which is $0.429$.

\begin{figure*}[t!]
        \centering
        \begin{subfigure}[t]{0.46\textwidth}
                \includegraphics[scale=0.55]{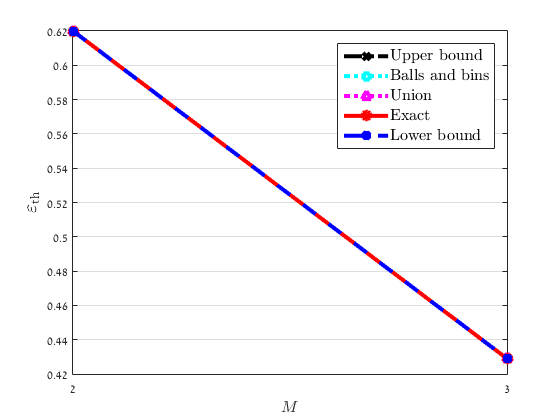}
                \caption{$q=3$}
                \label{th_q3}
        \end{subfigure}%
        ~ 
        \begin{subfigure}[t]{0.46\textwidth}
                \includegraphics[scale=0.55]{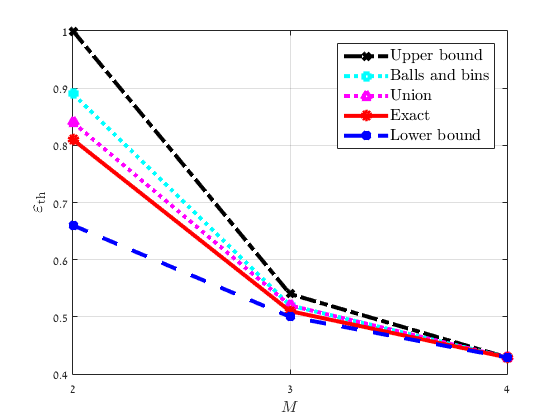}
                \caption{$q=4$}
                \label{th_q4}
        \end{subfigure}
        ~ 
        \begin{subfigure}[t]{0.46\textwidth}
                \includegraphics[scale=0.55]{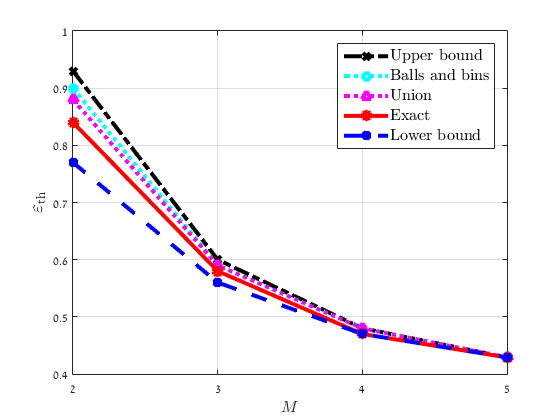}
                \caption{$q=5$}
                \label{th_q5}
        \end{subfigure}%
          ~ 
             \begin{subfigure}[t]{0.46\textwidth}
                \includegraphics[scale=0.55]{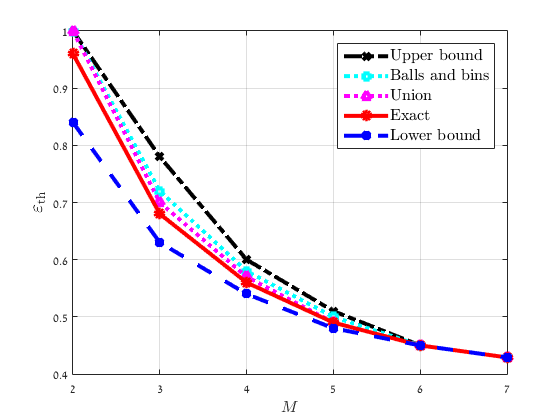}
                \caption{$q=7$}
                \label{th_q7}
        \end{subfigure}
              ~ 
           \begin{subfigure}[t]{0.46\textwidth}
                \includegraphics[scale=0.55]{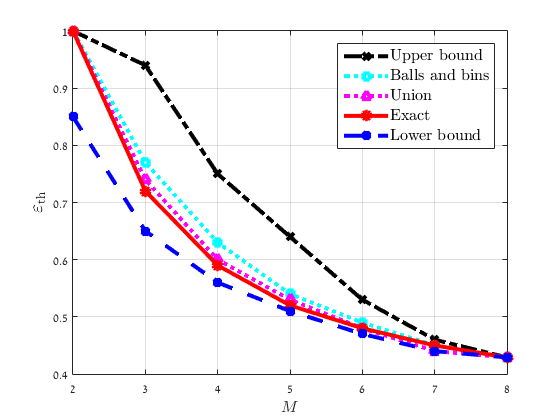}
                \caption{$q=8$}
                \label{th_q8}
        \end{subfigure}%
              ~ 
             \begin{subfigure}[t]{0.46\textwidth}
                \includegraphics[scale=0.55]{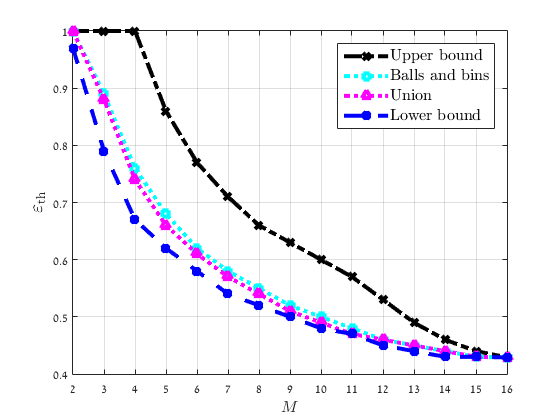}
                \caption{$q=16$}
                \label{th_q16}
        \end{subfigure}
        \caption{The QPEC threshold of the regular ($3,6$) LDPC code ensemble, as a function of $q$ and $M$. \label{fig_results}}
\end{figure*}

According to Figure \ref{fig_results}, the upper bound on the threshold calculated using the cardinality-based equations becomes loose as $q$ increases. This is due to the dependency of the sumset cardinality lower bound (see Theorem \ref{cor_bounds}) on the smallest prime factor of $q$, which is $2$ for binary fields. This makes the threshold upper bound for such fields less tight compared to prime fields. The lower bound is also somewhat loose, as it corresponds to the upper bound on the sumset cardinality that depends on the number of sums in the sumset. However, the bounds on the sumset cardinality are sharp (see Section \ref{pm_bounds}), so it is difficult to improve the bounds shown in Figure \ref{fig_results}. On the other hand, the balls-and-bins model and the union model provide good approximations of the exact threshold, and they are significantly tighter than the bounds. Recall that these approximations can be calculated efficiently, making the models especially attractive for large values of $q$. We deduce from Figure \ref{fig_results} an interesting result: not considering the algebraic structure of the field when using the approximation models leads on average to \textit{smaller} sumset cardinalities compared to the exact calculation of the susmset. If, as we conjecture, the approximation models give indeed upper bounds on the threshold, the uncertainty interval of the exact threshold is relatively small, and it becomes smaller as $M$ approaches $q$.

Figure \ref{fig_results} suggests a potential application of the QPEC to speed up the read process in measurement channels. As an example, suppose that $q=8$ and $M=4$. The decoding threshold in this case is approximately $0.59$. Thus, instead of performing $q-1=7$ comparative measurements to completely read the stored symbol, in $59\%$ of the cells we may perform only one measurement (yielding $q/2=4=M$ uncertainty). In terms of read rate, we now need only $3.46$ measurements on average, improving the read rate by more than $50\%$.

\section{Code Design using Linear Programming}
\label{code_design_considerations}

The design of good LDPC codes for the QPEC using the exact density-evolution equations \eqref{exact_de1}-\eqref{exact_de2} is difficult due to their $\mathcal{O}(2^{2q})$ dimensionality (see Section \ref{de_analysis}). Motivated by the efficiency and the good approximations obtained using the cardinality-based approach, we propose two methods for linear programming (LP) optimization of degree distributions. In Section~\ref{upper_bound_dt}, we present a threshold-oriented iterative optimization process. In Section~\ref{code_design_union}, we use the union model to achieve a target of small decoding-failure probability.

\subsection{Iterative QPEC code design}
\label{upper_bound_dt}

Let us denote by $\varepsilon_{\rm th}^{\rm c}$ the approximate decoding threshold obtained using the cardinality-based equations \eqref{cardinality_de1}-\eqref{cardinality_de2}. In this subsection, we propose two LP optimization methods for obtaining a degree-distribution pair with a desired $\varepsilon_{\rm th}^{\rm c}$ value. Recall that $W_q^{\left( l \right)}$ denotes the probability of a CTV message of cardinality $q$ at iteration $l$ of the decoding process when the cardinality-based equations are used (see Section \ref{de_equations}). Assuming a QPEC with $M>q/2$, a sufficient condition for an outgoing CTV message to be of cardinality $q$ is the $q$-condition, meaning that there is at least one pair of incoming VTC messages whose sum of cardinalities exceeds $q$ (see Lemma \ref{sct}). Therefore, $W_q^{\left( l \right)}$ is bounded from below by the probability that at least two incoming VTC messages are of cardinality $M$, since $2M>q$ when $M > q/2$. As $Z_M^{(l-1)}$ denotes the probability for a VTC message of cardinality $M$ at iteration $l-1$, we get:
\begin{align}
\label{wq}
 W_q^{\left( l \right)} & \ge \sum\limits_{i=2}^{d_c} {{\rho _i}} \sum\limits_{j = 2}^{i - 1} {{i-1 \choose j} {{\left( {Z_M^{\left( {l - 1} \right)}} \right)}^j}{{\left( {1 - Z_M^{\left( {l - 1} \right)}} \right)}^{i - 1 - j}}} \\ \nonumber &=  1 - \rho \left( {1 - {Z_M^{\left( {l - 1} \right)}}} \right) - {Z_M^{\left( {l - 1} \right)}}\rho '\left( {1 - {Z_M^{\left( {l - 1} \right)}}} \right), \nonumber 
\end{align}
where $\rho'(x)$ denotes the derivative of the polynomial $\rho(x)$ with respect to $x$. A sufficient condition for obtaining VTC messages of cardinality $M$ is that a variable node is a partial erasure and all its incoming CTV messages are of cardinality $q$. Therefore,
\begin{equation}
\label{zM}
{Z_M^{\left( l-1 \right)}} \ge \varepsilon \sum\limits_{i=1}^{d_v} {{\lambda _i}} {\left( {W_q^{\left( l-1 \right)}} \right)^{i - 1}} = \varepsilon \lambda \left( {W_q^{\left( l-1\right)}} \right).
\end{equation}
$\lambda(x)$ is an increasing function of $x$ for $x \ge 0$, since $\lambda(x)$ is a polynomial with non-negative coefficients. Note that both $W_q^{(l)}$ and the right-hand side of \eqref{wq} are non-negative as probabilities. Thus, according to \eqref{wq},
\begin{align}
\label{lambda_ge}
\lambda \left( {W_q^{\left( {l - 1} \right)}} \right) \ge \lambda \left( {1 - \rho \left( {1 - Z_M^{\left( {l - 2} \right)}} \right) -  Z_M^{\left( {l - 2} \right)}\rho '\left( {1 - Z_M^{\left( {l - 2} \right)}} \right)} \right).
\end{align}
Finally, by combining \eqref{zM} and \eqref{lambda_ge}, we get:
\begin{equation}
\label{recurrence_ineq}
{Z_M^{\left( l \right)}} \ge \varepsilon \lambda \left( {1 - \rho \left( {1 - {Z_M^{\left( {l - 1} \right)}}} \right) - {Z_M^{\left( {l - 1} \right)}} \rho '\left( {1 - {Z_M^{\left( {l - 1} \right)}}} \right)} \right),
\end{equation}
with the initial condition $Z_M^{(0)}=\varepsilon$.

The inequality in \eqref{recurrence_ineq} applies to any $M > q/2$, for all $q$ (which is prime or prime power), and depends solely on the degree distributions $\lambda(x)$ and $\rho(x)$. Recall that in the case of the BEC, we have an equality rather than an inequality, and without the additional term ${ - Z_M^{\left( {l - 1} \right)}\cdot\rho '\left( {1 - Z_M^{\left( {l - 1} \right)}} \right)}$. This term leads to an upper bound on $\varepsilon_{\rm th}^{\rm c}$, as we will see. Define the function $h_\varepsilon(x) = \varepsilon \lambda \left( {1 - \rho \left( {1 - x} \right) - x\rho '\left( {1 - x} \right)} \right)$ (which is the right-hand side of the inequality in \eqref{recurrence_ineq} with $Z_M^{(l-1)}$ replaced by $x$), and denote by ${h_\varepsilon^l}\left( x \right)$ the $l^{\text{th}}$ composition of $h_\varepsilon(x)$ with itself. We begin with the following lemma.

\begin{lemma}
\label{lem:he}
\begin{enumerate}
\item $Z_M^{\left( l \right)} \ge {h_\varepsilon^l}\left( \varepsilon \right), \hspace{0.2in} l \ge 1$.
\item $\mathop {\lim }\limits_{l \to \infty } h_\varepsilon ^l\left( \varepsilon \right)$ exists and is an increasing function of $\varepsilon$.
\end{enumerate}
\end{lemma}
The proof of this lemma is provided in Appendix \ref{proof:lem_he}. Observing that $h_{\varepsilon}^l\left( \varepsilon \right)=0$ for $\varepsilon=0$ and using the second part of Lemma \ref{lem:he}, we are able to define the following value:
\begin{equation}
\label{epsilon_star_definition}
{\varepsilon ^*} = \sup \left\{ {\varepsilon  \in \left[ {0,1} \right]:\mathop {\lim }\limits_{l \to \infty } h_\varepsilon ^l\left( \varepsilon \right) = 0} \right\}.
\end{equation}
Note that ${\varepsilon ^*}$ is defined with respect to a certain degree-distribution pair $\lambda$ and $\rho$. This definition of ${\varepsilon ^*}$ leads to an upper bound on $\varepsilon_{\rm th}^{\rm c}$.

\begin{theorem}
\label{upper_bound_thm}
For a QPEC with $M>q/2$, $\varepsilon_{\rm th}^{\rm c} \le {\varepsilon ^*}$.
\end{theorem}
\begin{proof}
$Z_M^{\left( l \right)}$ is bounded from below by a strictly positive value for all $l$ when $\varepsilon  > {\varepsilon^*}$, according to Lemma \ref{lem:he} and the definition of $\varepsilon^*$ in \eqref{epsilon_star_definition}. Since the probability of decoding failure according to the cardinality-based approach \eqref{Perror_cardinality} in this case is necessarily non-zero, $\varepsilon_{\rm th}^{\rm c}$ cannot exceed ${\varepsilon ^*}$. \qed
\end{proof}
For the formulation of an LP optimization, we derive an equivalent definition for $\varepsilon^*$, by extending the fixed-point characterization of the BEC threshold \cite{Richardson1}.
\begin{theorem}
\label{thm:epsilon_star_fp} For a QPEC with $M>q/2$,
\begin{equation}
\label{epsilon_star_equivalent_def}
{\varepsilon ^*} = \sup \bigg\{ {\varepsilon  \in \left[ {0,1} \right]:x = {h_\varepsilon }\left( x \right){\text{ has no solution }}x{\text{ in }}\left( {0,1} \right]} \bigg\}.
\end{equation}
\end{theorem}
The proof of this theorem is similar to the proof of Theorem 3.59 in \cite{mct}, and is omitted. We now formulate an LP optimization for determining good (in terms of code rate) variable-node degree distribution $\lambda(x)$ for given $\rho(x)$ and $\varepsilon^*$ assuming that $M>q/2$. A maximum constraint $d_v$ on variable-node degrees is set, as usual~\cite{mct}, to control implementation complexity and convergence speed. According to the $\varepsilon^*$ equivalent definition \eqref{epsilon_star_equivalent_def}, the condition for degree distributions whose threshold is upper bounded by $\varepsilon^*$ is that $h_{\varepsilon^*}(x) - x \le 0$ for $x \in \left( {0,1} \right]$. This leads us to formulate an LP optimization for the QPEC, where maximal rate is sought under the constraint that $\varepsilon_{\rm th}^{\rm c}$ is upper bounded by $\varepsilon^*$:
\begin{align}
\label{opt_problem1}
{\max _\lambda } \bigg\{ \sum\limits_{i = 2}^{{d_v}} \frac{{{\lambda _i}}}{i}:&{\lambda _i} \ge 0,\sum\limits_{i = 2}^{{d_v}} {{\lambda _i}}  = 1,h_{\varepsilon^*}(x) - x\le 0, x \in \left( {0,1} \right] \bigg\}.
\end{align}
We term the LP optimization in \eqref{opt_problem1} as QPEC* LP. Note that the decoding threshold increases as $M$ decreases, such that the degree distributions obtained by QPEC* LP provide at least the same threshold for a QPEC with $M \le q/2$. The difference between the known BEC (or QEC) LP \cite{mct} and QPEC* LP is in using in \eqref{opt_problem1} the function $h_{\varepsilon^*}(x)$ specially developed for the QPEC, instead of the function ${f_\varepsilon }(x) = \varepsilon \cdot \lambda \left( {1 - \rho \left( {1 - x} \right)} \right)$ derived from the BEC density-evolution equation. 

The QPEC* LP optimization provides a degree-distribution pair with $\varepsilon_{\rm th}^{\rm c}$ upper-bounded by $\varepsilon^{*}$. This suggests the following strategy for obtaining degree distributions with a desired value of $\varepsilon_{\rm th}^{\rm c}$. Choose $\varepsilon^*$ that is larger than the desired $\varepsilon_{\rm th}^{\rm c}$, and solve the QPEC* LP optimization. Find $\varepsilon_{\rm th}^{\rm c}$ of the optimized degree distributions using the cardinality-based density-evolution equations (where the union model is suggested for large $q$). If the threshold is smaller than $\varepsilon_{\rm th}^{\rm c}$, increase $\varepsilon^{*}$ and repeat the process. Otherwise, decrease $\varepsilon^{*}$ and repeat the process. An alternative design method using previously known theoretical tools is to seek the desired $\varepsilon_{\rm th}^{\rm c}$ using the BEC LP optimization. Because the BEC is a degraded version of the QPEC, here we will choose a target BEC threshold $\varepsilon_{\rm th}^{\rm BEC}$ smaller than the desired $\varepsilon_{\rm th}^{\rm c}$. We then similarly calculate $\varepsilon_{\rm th}^{\rm c}$ of the resulting degree distributions, and decrease/increase the BEC threshold as needed (note that the BEC LP approach is valid for $M \le q/2$ as well).

It turns out that using the QPEC* LP approach can result in better codes compared to the BEC optimization. In the sequel we show this by numerical examples. The intuition behind this improvement is that the QPEC* LP optimization better captures the decoding performance for QPECs with $M<q$. We now show the benefit of the new QPEC LP optimization in achieving better code ensembles than those obtained using the BEC LP optimization. As an example, assume that $\rho(x)=x^5$, $d_v=5$ and the desired $\varepsilon_{\rm th}^{\rm c}$ is $0.6$. We concentrate here on QPECs with $M = \left\lfloor {q/2} \right\rfloor  + 1$ for several values of $q$ (this value of $M$ is the smallest satisfying $M>q/2$). An illustration of the iterative optimization process is provided in Figure \ref{opt_results}. The plot shows the sequence of optimization runs of the QPEC* optimizer (right), and the sequence of runs for the BEC optimizer (left). The QPEC* LP approaches the target of $\varepsilon_{\rm th}^{\rm c} = 0.6$ from above, and the BEC LP from below. Note the approximate linear behaviour of $\varepsilon_{\rm th}^{\rm c}$ as a function of $\varepsilon^*$, rendering the iterated QPEC* LP as a simpler way for code design. As a consequence, reaching the desired QPEC threshold took typically fewer optimization instances with the QPEC* optimizer than with the BEC optimizer.   

\begin{figure}[t!]
\centering
\includegraphics[scale=0.55]{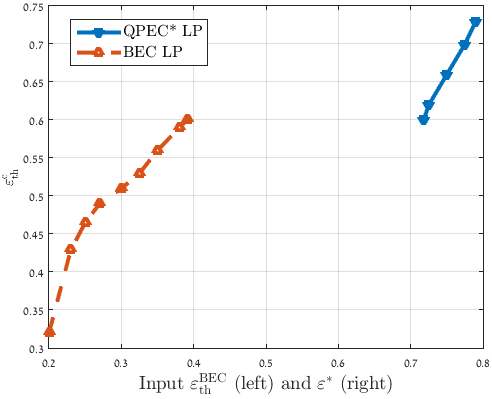}
\caption{An example of the iterative optimization process ($q=3, M=2$).}
\label{opt_results}
\end{figure}
\begin{figure*}[t!]
\begin{center}
\captionof{table}{Iterative optimization results for $\rho(x)=x^5$, $d_v=5$ ($\varepsilon_{\rm th}^{\rm c} = 0.6$).} \label{tab:opt_results_QPEC} \medskip
    \begin{tabular}{| c | c || c | c | c || c | c | c|}
    \hline
    $q$ & $M$ & QPEC* LP $\lambda(x)$ \Tstrut \Bstrut & Rate & $\varepsilon^*$ & BEC LP $\lambda(x)$ & Rate & $\varepsilon_{\rm th}^{\rm BEC}$ \\ \hline \Tstrut
    $3$ & $2$ & $0.644x + 0.356{x^4}$ & $0.576$ & $0.718$ & $0.517x + 0.099x^2+0.384x^3$ & $0.569$ & $0.391$ \\ \hline \Tstrut
    
    $4$ & $3$ &$0.193x + 0.807{x^4}$ & $0.354$ & $0.778$  & $0.157x+0.843x^4$ & $0.325$ & $0.532$ \\ \hline \Tstrut

    $5$ & $3$ & $0.489x + 0.511{x^4}$ & $0.519$ & $0.751$ & $0.437x+0.056x^2 +0.507x^4$ & $0.508$ & $0.464$ \\ \hline \Tstrut
    
    $7$ & $4$ & $0.372x + 0.628{x^4}$  &  $0.465$ & $0.763$ & $0.345x+0.655x^4$ &  $0.451$ & $0.492$ \\ \hline \Tstrut
    
    $8$ & $5$ & $0.46x + 0.54x^4$  & $0.507$ & $0.749$ & $0.413x+0.587x^4$ &  $0.485$ & $0.48$ \\ \hline \Tstrut
    
        $16$ & $9$ & $0.422x+0.578x^4$ & $0.489$ & $0.754$ & $0.385x + 0.615x^4$  &  $0.471$ & $0.487$ \\\hline

\end{tabular}

\end{center}
\end{figure*}

\begin{figure*}[t!]
    \centering
    \begin{subfigure}[t]{0.5\textwidth}
        \centering
        \includegraphics[scale=0.55]{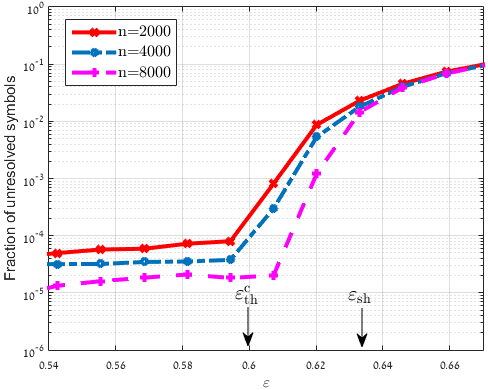}
        \caption{QPEC* LP.}
\label{qpec_lp}
    \end{subfigure}%
    ~ \hspace{5pt}
    \begin{subfigure}[t]{0.5\textwidth}
        \centering
        \includegraphics[scale=0.55]{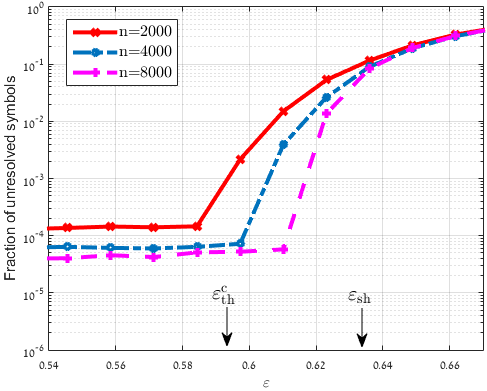}
        \caption{BEC LP.}
\label{bec_lp}
    \end{subfigure}
    \caption{Finite-length decoding performance of the LP optimized degree distributions designed for $q=8, M=5$ in Table \ref{tab:opt_results_QPEC} (rate $0.507$). $\varepsilon_{\rm sh}=0.637$ is the Shannon limit for these QPEC parameters with rate $0.507$.}
    \label{fig:LP_opt}
\end{figure*}

The optimized variable degree distributions and their corresponding rates are listed in Table \ref{tab:opt_results_QPEC}, together with the values of $\varepsilon^*$ and $\varepsilon^{\rm BEC}_{\rm th}$ resulting in $\varepsilon_{\rm th}^{\rm c}$. When comparing the results, we observe that for all the parameters checked the rates achieved by the QPEC* optimizer are strictly better than the rates resulting from the BEC optimizer. Another interesting observation is that in some cases the BEC optimizer required a $\lambda(x)$ polynomial with more non-zero coefficients than the QPEC optimizer. To evaluate the decoding performance in the practical setting of finite-length codes, we constructed random parity-check matrices for varying code lengths and performed $80$ iterative decoding iterations using the message-passing decoder described in Section \ref{message_passing}. In Figure \ref{fig:LP_opt}, we compare the performance of the QPEC* LP and BEC LP optimized degree distributions for $q=8$ and $M=5$. For a fair comparison, we set $\varepsilon_{\rm th}^{\rm BEC}$ to $0.472$ to obtain a degree-distribution pair with the same rate ($0.507$) as in the QPEC* LP optimization, leading to $\varepsilon_{\rm th}^{\rm c}=0.594$ in the BEC LP optimization. The QPEC* optimized degree-distribution pair exhibits notably superior decoding performance, for complexity similar to the BEC LP. That is, the QPEC* LP, tailored for the QPEC, better captures the channel behaviour compared to the use of the BEC LP.

\subsection{Code design using the union model}
\label{code_design_union}

\begin{figure*}[t]
\begin{center}
\captionof{table}{Optimized variable degree distributions (the code rates are approximately $1/2$).}\label{tab:QPEC_lp_union} \medskip 
    \begin{tabular}{| c | c | c |}
    \hline \Tstrut
        $q$ & $M$ & Optimized $\lambda(x)$ \\  \hline \Tstrut
    $4$ & $2$ & $0.0522x + 0.0799{x^2} + 0.6983{x^3} + 0.1668{x^4} + 0.0028{x^6}$  \\  \hline  \Tstrut
      $8$ & $2$ & $0.0394x + 0.0298{x^2} + 0.9007{x^3} + 0.0099{x^4} + 0.0142{x^5} + 0.005{x^7} + 0.001{x^8}$  \\  \hline \Tstrut
        $8$ & $4$ & $0.0813x + 0.0402{x^2} + 0.7717{x^3} + 0.0196{x^4} + 0.082{x^5} + 0.0025{x^6} + 0.0007{x^7} + 0.002{x^9}$   \\  \hline  \Tstrut
          $16$ & $4$ & $0.0526x + 0.0584{x^2} + 0.8237{x^3} + 0.0095{x^4} + 0.0104{x^5} + 0.044{x^6} + 0.0014{x^7}$ \\  \hline     \Tstrut        
              $16$ & $8$ & $0.0226x + 0.2002{x^2} + 0.6592{x^3} + 0.0593{x^4} + 0.0225{x^5} + 0.0128{x^6} + 0.0091{x^7} + 0.0075{x^8} + 0.0068{x^9}$   \\  \hline   
                         
\end{tabular}
\medskip
\end{center}
\end{figure*}

\begin{figure*}[t]
    \centering
    \begin{subfigure}[t]{0.5\textwidth}
        \centering
        \includegraphics[scale=0.55]{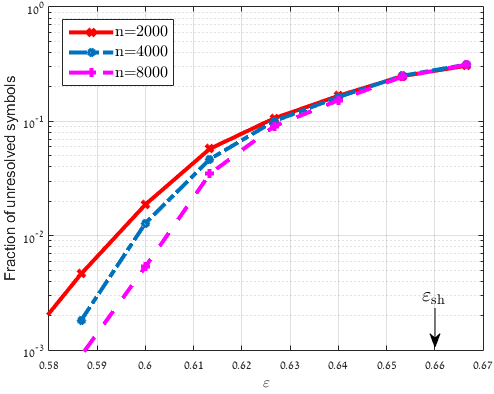}
        \caption{Union-optimized degree distributions.}
\label{union_finite}
    \end{subfigure}%
    ~ \hspace{5pt}
    \begin{subfigure}[t]{0.5\textwidth}
        \centering
        \includegraphics[scale=0.55]{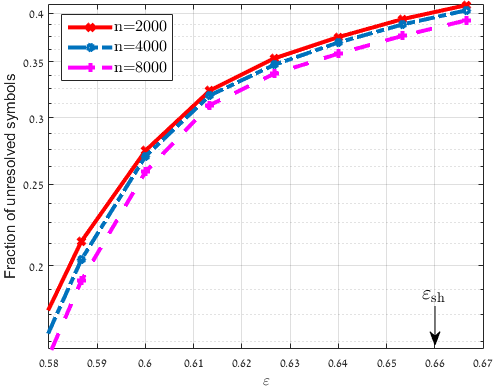}
        \caption{BEC LP optimized degree distributions.}
\label{union_finite_bec}
    \end{subfigure}
    \caption{Finite-length decoding performance of the optimized degree distributions for $q=16, M=8$ in Table \ref{tab:QPEC_lp_union} (rate $1/2$). The BEC LP results are shown for comparison. $\varepsilon_{\rm sh}=0.66$ is the Shannon limit for these QPEC parameters with rate $1/2$.}
    \label{fig:union_design}
\end{figure*}

In this part, we use the union model to obtain degree-distribution pairs that achieve a small probability of decoding failure. We begin with a given degree-distribution pair $\left( {\lambda (x),\rho (x)} \right)$ and a desired (small) probability of decoding failure $p_{\rm tar}$, such that there exists an iteration number $L$ satisfying ${p_e^{(L)}} \le {p_{{\rm{tar}}}} < {p_e^{(L - 1)}}$. Our aim is to find $\tilde \lambda(x)$ that either achieves a lower target of decoding-failure probability, or achieves the same probability in fewer decoding iterations. To find such a variable degree distribution, we adapt the optimization method suggested in \cite{Richardson1} to the QPEC, as follows. Define ${A}_{l,i}$ as the probability of decoding failure at iteration $l$ assuming that we use $\lambda(x)$ at the first $l-1$ iterations followed by the use of the node variable degree distributions with its mass on the degree $i$ at iteration $l$. Note that $p_e^{(l)}$ is obtained as the following sum
\begin{equation}
\label{pe_lambda}
p_e^{(l)} = \sum\limits_{i = 2}^{{d_v}} {{{A}_{l,i}} \cdot {\lambda _i}}.
\end{equation}
We also define the probability of decoding failure due to degree-$i$ variable nodes assuming that $\lambda(x)$ is used for the first $l-1$ iterations and that $\tilde \lambda (x)$ is used at iteration $l$ as ${\tilde p}_e^{(l)}$, which is the right-hand size of \eqref{pe_lambda} with $\lambda_i$ replaced with $\tilde \lambda _i$
The number of decoding iterations required to obtain the probability of decoding failure $p_e^{(L)}$ is approximated as \cite{Richardson1}
\begin{equation}
\label{exp_iterations}
G\left( \tilde \lambda  \right) \simeq \sum\limits_{l = 1}^L {\frac{{{{\tilde p_e}^{(l)}} - {p_e^{(l)}}}}{{{p_e^{(l - 1)}} - {p_e^{{(l)}}}}}}, 
\end{equation}
assuming that $\tilde \lambda(x)$ and $\lambda(x)$ do not differ much. Finally, the following LP optimization is obtained
\begin{align}
\label{lp1}
{\min _{\tilde \lambda }} \bigg\{ G\left( {\tilde \lambda } \right):&{{\tilde \lambda }_i} \ge 0,\sum\limits_{i = 2}^{{d_v}} {{{\tilde \lambda }_i}}  = 1,{{\tilde p}_e^{(l)}} \le {p_e^{(l - 1)}},
\\\nonumber
&{{\max }_l}\frac{{\left| {{{\tilde p}_e^{(l)}} - {p_e^{(l)}}} \right|}}{{{p_e^{(l - 1)}} - {p_e^{(l)}}}} \le \delta \bigg\} ,
\end{align}
for $1 \le l \le L$ and $\delta  \ll 1$. That is, the LP optimization in \eqref{lp1} searches for a perturbed version of $\lambda(x)$ with better performance, i.e., with a smaller number of decoding iterations resulting in the desired probability of decoding failure. The LP optimization is repeated with the optimized $\tilde \lambda (x)$ as an input until convergence is achieved.

The major difficulty in solving the LP optimization in \eqref{lp1} in the QPEC case is the need to calculate $p_e^{(l)}$. The calculation of this probability is difficult even when the set cardinality approach is taken, as we saw in Section \ref{de_equations}. Therefore, we calculate this probability under the union model, motivated by its good approximation of the exact behaviour (see Section \ref{subsection:comparison}). We concentrate on an initial $\rho(x)$ with only two non-zero consecutive degrees, which usually provide good results \cite{Richardson1, chung_thesis} and limit the search space. To control complexity, the maximum variable node degree is set to $d_v=10$. The initial $\lambda(x)$ was chosen such that the number of non-zero degrees is small \cite{chung_thesis} and the initial rate is close to $1/2$. Table \ref{tab:QPEC_lp_union} summarizes the optimization results for several values of $M$ and $q$, where $\rho \left( x \right) = 0.4 \cdot {x^6} + 0.6 \cdot {x^7}$ is the initial check degree distribution.

For comparison, we used the known BEC LP optimization \cite{mct} to obtain a good degree distribution with rate $1/2$ (as the code rates in Table \ref{tab:QPEC_lp_union}). The BEC optimized variable degree-distribution pair is ${\lambda _{{\rm{BEC}}}}\left( x \right) = 0.3195x + 0.1554{x^2} + 0.5251{x^9}$, obtained by setting  $\rho \left( x \right) = 0.4 \cdot {x^6} + 0.6 \cdot {x^7}$ and the BEC threshold $0.473$ in the BEC LP. In Figure \ref{fig:union_design}, we present the average finite-length decoding performance of the optimized degree distributions obtained in Table \ref{tab:QPEC_lp_union} for $q=16$ and $M=8$ compared to the BEC LP optimized degree-distribution pair. As expected, the QPEC-designated optimization \eqref{lp1} provides significantly better results compared to simply using the BEC LP. This is explained by the strict requirement in the BEC LP to recover from full erasures, whereas only $M=8$ partial erasures should be considered.

\section{Conclusion}
\label{conc}

This work offers a study of the performance of iterative decoding of GF($q$) LDPC codes over the newly defined QPEC. Generalizing the BEC to partial erasures, the QPEC serves as a useful model for partial data loss. We extended the BEC decoder to deal with partial erasures, and demonstrated the spectrum of possible messages in the QPEC decoding process. As a consequence, we showed that the QPEC, unlike the BEC, introduces non-trivial challenges in performance analysis. Therefore, we developed efficient approximation models, which provide important tools for the QPEC decoding performance evaluation. We also suggested LP optimizations for finding LDPC codes with good decoding performance, which provided better results compared to the known BEC LP optimization.

The QPEC model is an initial step in the analysis of measurement channels. These channels and the concept of partial erasures encourage the development of additional models and efficient analysis methods. As a future research, it is suggested to investigate the relation between the algebraic operations of the field and the models proposed in this work. In addition, apart from the random sets provided as the QPEC output, one may consider a model in which the output sets are structured (e.g., contain consecutive levels).

\ifCLASSOPTIONcaptionsoff
  \newpage
\fi

\appendices

\section{Proof of Theorem \ref{QPEC_cap}}
\label{proof:uniform_dist}
Define the sets $\left\{ {{\mathcal{A}_i}} \right\}_{i = 1}^{T}$, each containing $M$ elements taken from $\mathcal{X}$, such that $\mathcal{A}_i \ne \mathcal{A}_{j}$ for $i \ne j$ and $T={q \choose M}$. The output symbol $Y$ is a set of either one symbol or $M$ symbols. Its entropy given an input distribution ${\left\{ {{p_x}} \right\}}$ (for $x=0,\alpha^0,\alpha^1,...,\alpha^{q-2}$) is:
\begin{align}
\label{eq:Hy_px}
&H\left( {Y;\left\{ {{p_x}} \right\}} \right)  =  - \sum\limits_{x \in \mathcal{X}} {p_x(1-\varepsilon)\log \left( {{p_x}(1-\varepsilon)} \right)} {\text{ }} \nonumber \\ &- \sum\limits_{i= 1}^T {\left( {\frac{\varepsilon}{{q-1 \choose M-1}}\sum\limits_{x \in {\mathcal{A}_i}} {{p_x}} } \right)\log \left( {\frac{\varepsilon}{{q-1 \choose M-1}}\sum\limits_{x \in {\mathcal{A}_i}} {{p_x}} } \right)}.
\end{align}
The capacity achieving distribution can be found by solving the following maximization problem:
\begin{equation}
{\max _{\left\{ {{p_x}} \right\}}}\hspace{2pt}H\left( {Y;\left\{ {{p_x}} \right\}} \right),\hspace{7pt} {\text{s}}.{\text{t}}. \hspace{2pt} \sum\limits_{x \in \mathcal{X}} {{p_x}}  = 1.
\end{equation}
Using the method of Lagrange multipliers, we get the following system of equations:
\begin{align}
&\frac{{\partial H\left( {Y;\left\{ {{p_x}} \right\}} \right)}}{{\partial {p_x}}} + \lambda  = 0, {\hspace{3pt} \rm for \hspace{3pt}} x=0,\alpha^0,\alpha^1,...,\alpha^{q-2}
\\\nonumber
&{\rm s.t.} \hspace{3pt} \sum\limits_{x \in \mathcal{X}} {{p_x}}  = 1,
\end{align}
where $\lambda$ is the Lagrange multiplier. These equations translate into:
\begin{equation}
\begin{array}{l}
 - (1 - \varepsilon )\left( {\log {p_x} + 1} \right) - \sum\limits_{{\mathcal{A}_i}:x \in {\mathcal{A}_i}} {{\frac{\varepsilon}{{q-1 \choose M-1}}}\log \left( {{\frac{\varepsilon}{{q-1 \choose M-1}}}\sum\limits_{x \in {\mathcal{A}_i}} {{p_x}} } \right)}  \\
 - {\frac{\varepsilon}{{q-1 \choose M-1}}}\left( {T - 1} \right) + \lambda  = 0,\hspace{7pt} \sum\limits_{x \in \mathcal{X}} {{p_x}}  = 1,
\end{array}
\end{equation}
which are satisfied if $p_x=1/q$ for all $x \in \mathcal{X}$. The mutual information $I\left( {X;Y} \right)$ is a concave function of $p_x$, and therefore $p_x=1/q$ leads to the global maximum of \eqref{cap_def}, that is, to the capacity. Finally, the QPEC capacity \eqref{QPEC_capacity} is obtained by substituting \eqref{eq:Hy_px} and \eqref{cond_entropy} in \eqref{cap_def} and setting $p_x=1/q$ for all $x$. \qed

\section{Proof of Lemma \ref{lem:he}}
\label{proof:lem_he}

We begin by proving that $h_\varepsilon(x)$ is an increasing function of both $\varepsilon$ and $x$, for $\varepsilon, x  \in \left[ {0,1} \right]$, by taking the partial derivatives of $f\left( {\varepsilon ,x} \right)$ with respect to $\varepsilon$ and $x$:
\begin{equation}
\label{diff_e}
\frac{{\partial f}}{{\partial \varepsilon }} = \lambda \left( {1 - \rho \left( {1 - x} \right) - x\rho '\left( {1 - x} \right)} \right),
\end{equation}
\begin{equation}
\label{diff_x}
\frac{{\partial f}}{{\partial x}} = \varepsilon x\lambda '\left( {1 - \rho \left( {1 - x} \right) - x\rho '\left( {1 - x} \right)} \right)\rho ''\left( {1 - x} \right),
\end{equation}
where $\rho'(x)$ and $\rho''(x)$ denote the first and second derivatives of $\rho(x)$. The polynomials $\rho \left( x \right), \lambda \left( x \right)$ and their derivatives are power series of $x$ with non-negative coefficients, and as such they are non-negative for $x \ge 0$. In particular, $\rho ''\left( {1 - x} \right) \ge 0$ since $0 \le 1- x \le 1$. Therefore, it is sufficient to prove that $g(x) = { \rho \left( {1 - x} \right) + x\rho '\left( {1 - x} \right)} \le 1$ for establishing the non-negativity of the partial derivatives \eqref{diff_e} and \eqref{diff_x}. This is proved in the following manner. The derivative of $g(x)$ in the interval $(0,1)$ satisfies $g'\left( x \right) =  - x\rho ''\left( {1 - x} \right) < 0$, meaning that $g(x)$ is a decreasing function of $x$. In particular, $g\left( x \right) \le g\left( 0 \right) = \rho \left( 1 \right) = 1$, as needed. Now, $z_M^{\left( l \right)} \ge h_\varepsilon\left( {z_M^{\left( {l - 1} \right)}} \right)$ according to the inequality in \eqref{recurrence_ineq}. Thus,
\begin{equation}
\label{ineq_mon}
h_\varepsilon\left( {z_M^{\left( {l - 1} \right)}} \right) \ge h_\varepsilon\left( {h_\varepsilon\left( {z_M^{\left( {l - 2} \right)}} \right)} \right), \hspace{0.2in} l \ge 2.
\end{equation}
As we saw earlier, $h_\varepsilon(x)$ is an increasing function of $x$. Repeated application of the monotonicity property to the right-hand side of \eqref{ineq_mon} leads to the inequality ${h_\varepsilon }\left( {z_M^{\left( {l - 1} \right)}} \right) \ge h_\varepsilon ^l\left( {z_M^{\left( 0 \right)}} \right)$, where ${h_\varepsilon^l}\left( x \right)$ denotes the $l^{\text{th}}$ composition of $h_\varepsilon(x)$ with itself. Therefore, $z_M^{\left( l \right)} \ge {h_\varepsilon }\left( {z_M^{\left( {l - 1} \right)}} \right) \ge h_\varepsilon ^l\left( {z_M^{\left( 0 \right)}} \right)= h_\varepsilon ^l\left( \varepsilon \right)$, proving the first part of the theorem. The second part of the theorem is proved using the monotonicity property of $h_\varepsilon(x)$ proved above and similar arguments to those used in Section 3.10 of \cite{mct} (where the BEC is considered). \qed

\bibliographystyle{IEEEtran}
	\bibliography{QPEC_paper_bib}

\begin{IEEEbiographynophoto}{Rami Cohen} (S'12) received the B.Sc. in Electrical Engineering and Physics (\textit{cum laude}) and the M.Sc. in Electrical Engineering from the Technion - Israel Institute of Technology in 2010 and 2012, respectively. He is currently pursuing the Ph.D. degree at the Department of Electrical Engineering, Technion - Israel Institute of Technology. His current research interests lie in coding theory and information theory, in particular the analysis and design of codes for high-speed memory devices and systems.
\end{IEEEbiographynophoto}

\begin{IEEEbiographynophoto}{Yuval Cassuto} (S'02-M'08-SM'14) is a faculty member at the Department of Electrical Engineering, Technion -- Israel Institute of Technology. His research interests lie at the intersection of the theoretical infomration sciences and the engineering of practical computing and storage systems.

During 2010-2011 he has been a Scientist at EPFL, the Swiss Federal Institute of Technology in Lausanne. From 2008 to 2010 he was a Research Staff Member at Hitachi Global Storage Technologies, San Jose Research Center. From 2000 to 2002, he was with Qualcomm, Israel R\&D Center, where he worked on modeling, design and analysis in wireless communications.

He received the B.Sc degree in Electrical Engineering, summa cum laude, from the Technion, Israel Institute of Technology, in 2001, and the MS and Ph.D degrees in Electrical Engineering from the California Institute of Technology, in 2004 and 2008, respectively.

Dr. Cassuto has won the 2010 Best Student Paper Award in data storage from the IEEE Communications Society, as well as the 2001 Texas Instruments DSP and Analog Challenge \$100,000 prize.
\end{IEEEbiographynophoto}

\end{document}